 \documentclass[final,3p,twocolumn,times]{elsarticle}

\usepackage{ifluatex,ifxetex}
\ifluatex
  \usepackage{fontspec} 
\else\ifxetex
  \usepackage{fontspec} 
\else 
  \usepackage[T1]{fontenc}
  \usepackage[utf8]{inputenc} 
\fi\fi

\usepackage{graphics} 
\usepackage{epsfig} 
\usepackage{amsmath} 
\usepackage{amsthm} 
\usepackage{amssymb}  
\usepackage{color,xcolor}
\usepackage{multirow}
\usepackage{graphicx}
\usepackage{xcolor}
\usepackage{subfig}
\usepackage{epstopdf}
\usepackage[ruled,vlined]{algorithm2e}

\usepackage{mathtools,float}
\usepackage{enumerate}
\usepackage{balance}

\newtheorem{definition}{Definition}
\newtheorem{theorem}{Theorem}
\newtheorem{lemma}{Lemma}
\newtheorem{remark}{Remark}
\newtheorem{assumption}{Assumption}
\newtheorem{stassumption}[assumption]{Standing Assumption}
\newtheorem{proposition}{Proposition}

\newtheorem{example}{Example}
\newcommand{\argmin}{\mathop{\mathrm{argmin}}\limits}

\newcommand{\Id}{ \textrm{Id}}
\newcommand{\bR} { {\mathbb R}}
\newcommand{\bN} { {\mathbb N}}

\newcommand{\fix} { {\mathrm{fix}}}
\newcommand{\zer} { {\mathrm{zer}}}
\newcommand{\diag} { {\mathrm{diag}}}

\newcommand{\ca}[1]{\mathcal{#1}}

\newcommand{\bld}[1]{\boldsymbol{#1}}
\newcommand{\1}{\bld 1}
\newcommand{\0}{\bld 0}
\newcommand{\col }{\mathrm{col}}
\newcommand{\iffs}{\Leftrightarrow}
\newcommand{\QEDopen}{$\square$}
\newcommand{\QED}{$\blacksquare$}

\journal{European Journal of Control}

\begin{document}

\begin{frontmatter}



\title{\LARGE \bf An asynchronous distributed and scalable generalized Nash equilibrium seeking algorithm for strongly monotone games}

\author[CC]{Carlo Cenedese}
\author[GB]{Giuseppe Belgioioso}
\author[SG]{Sergio Grammatico}
\author[CC]{Ming Cao}
\address[CC]{Engineering and Technology Institute Groningen (ENTEG), University of Groningen, The Netherlands}
\address[GB]{Control System group, TU Eindhoven, Eindhoven, The Netherlands}
\address[SG]{Delft Center for Systems and Control, TU Delft, The Netherlands}

\begin{abstract}

In this paper, we present three distributed algorithms to solve a class of generalized Nash equilibrium (GNE) seeking problems in strongly monotone games. The first one (SD-GENO) is based on synchronous updates of the agents, while the second and the third (AD-GEED and AD-GENO) represent asynchronous solutions that are robust to communication delays. AD-GENO can be seen as a refinement of AD-GEED, since it only requires node auxiliary variables, enhancing the scalability of the algorithm. Our main contribution is to  prove converge to a variational GNE of the game via  an operator-theoretic approach. 
Finally, we apply the algorithms to network Cournot games and  show how different activation sequences and delays affect convergence. We also compare the proposed algorithms to the only other in the literature (ADAGNES), and observe that AD-GENO outperforms the alternative.
\end{abstract}





\end{frontmatter}

                                                                                                                                                                                                                                                                                                                                                                                                                                                                                                                                                                                                                                                                                                                                                                                                                                                                                                                                                                                                                                                                                                                                                                                                                                                                                                                                                                                                                                                                                                                                                                                                                                                                                                                 
\section{Introduction}
\subsection{Motivation and literature overview}
In modern society, multi-agent network systems arise in several areas, leading to increasing research activities. When self-interested agents interact between each other, one of the best mathematical tools to study the emerging collective behavior is noncooperative game theory over networks. In fact, networked games emerges in several application domains, such as smart grids  \cite{dorfer:simpson-porco:bullo:16,parise:colombino:grammatico:lygeros:14}, social networks \cite{grammatico:18tcns} and distributed robotics \cite{martinez:bullo:cortes:frazzoli:07,cenedese:CDC2018:towards_TV_prox_dyn}.
In a game setup, the players (or agents) aim at minimizing a local and private cost function which represents their individual interest, and, at the same time, satisfy local and global constraints, limiting the possible decisions, or  strategies. 
Usually, there is a dependency of the cost and the constraints of a player from the strategies of a subset of other players, generically called ``neighbors''. Thus, each decision is influenced by some local information, which is typically exchanged with the neighbors. One popular notion of solution for these games is a collective equilibrium where no player benefits from unilaterally changing its strategy, see~\cite{facchinei:2007:GNE_problems}.
 
In \cite{grammatico:18tcns,yi:pavel:2019:operator_splitting_approach,belgioioso:grammatico:17cdc}, the authors focused on developing synchronous and distributed equilibrium seeking algorithms for noncooperative games, namely, the case in which all the agents update their strategies at the same time. Even though this assumption is quite common, it can lead to sever limitations in the case of heterogeneous agents in the game. For example, imagine to have two types of agents, divided in terms of good and bad performances, in a synchronous update scheme; the former must wait the latter before a new update can be carried on. In fact, this would produce  a bottleneck in the overall performance. To overcome this problem, we focus on developing asynchronous update rules. Moreover, it is known that asynchronicity can also speed up the convergence, facilitate the insertion of new agents in the network and even increase robustness w.r.t. communication faults, see~\cite{BERTSEKAS:1991:Survey_Asynch} and references therein.  

Among the very first works on asynchronous distributed optimization, the one of Bertsekas and Tsitsiklis in \cite{bertsekas:1989:parallel_optimization} stands out. From there onward, several authors elaborated on these ideas and produced novel results for convex optimization~\cite{recht:2011:hogwild,Combettes:2015:stoch_quadi_fejer,liu:2015:asynchronous_parallel_stoch_coord_desc,nedic:2011:asynchronous_broadcast-based_convex_opt}. 
In~\cite{yi:pavel:2019:asynch_distributed_GNE_w_partial_info}, Yi and Pavel developed an asynchronous algorithm to solve noncooperative generalized games subject to equality constraints. This result was enabled by the framework (ARock), recently introduced by Peng et al. in~\cite{peng2016arock}, that provides a wide range of asynchronous variations of the classical fixed point iterative algorithms. 

In this paper, we propose an asynchronous algorithm robust to delayed information  to solve noncooperative games subject to affine coupling constraints. Furthermore, to achieve a fully decentralized update rule, we rely only on node auxiliary variables, preserving the scalability in the case of a large number of agents. This result is a significant contribution, due to the technical challenges in the asynchronous implementation of the algorithm, addressed by carefully analyzing the influence of the delayed information on the dynamics of the auxiliary variables. Finally, we compare the proposed solution to the one in~\cite{yi:pavel:2019:asynch_distributed_GNE_w_partial_info}, for the case of a Cournot game, showing that our algorithm achieves faster convergence. A preliminary  and partial version of these results were presented in~\cite{cenedese:2019:ECC}.

\subsection{Organization of the paper}
In Section~\ref{sec:problem_formulation}, we formalize the problem setup and define the concept of \textit{variational} GNE. In  Section~\ref{sec:synch_case}, we derive the iterative algorithm for GNE seeking for the synchronous case, i.e., SD-GENO. Its asynchronous counterpart (AD-GEED), that  adopt edge auxiliary variables, is then presented in Section~\ref{sec:asynch_case}. The main result of the paper is presented in Section~\ref{sec:AD_GENO}, where we introduce AD-GENO. Section~\ref{sec:simulations} is dedicated to the simulation results and  to the comparison between the different algorithms performance. Section~\ref{sec:conclusion} ends the paper presenting the conclusions and the outlooks of this work.  
  
\section{Notation}
\label{sec:notations}
\subsection{Basic notation}
The set of real, positive, and non-negative numbers are denoted by $\mathbb{R}$, $\mathbb{R}_{>0}$, $\mathbb{R}_{\geq 0}$, respectively; $\overline{\bR}\coloneqq \bR\cup\{\infty\}$. The set of natural numbers is $\mathbb{N}$. For a square matrix $A \in \bR^{n\times n}$, its transpose is $A^\top$, $[A]_{i}$ is the $i$-th row of the matrix and  $[A]_{ij}$ represents the element in $i$-th row and  $j$-th column. $A\succ 0 $ ($A\succeq 0 $)  stands for a positive definite (semidefinite) matrix,  instead $>$ ($\geq$) describes element wise inequality. $A\otimes B$ is the Kronecker product of the matrices $A$ and $B$.   The identity matrix is denoted by~$I_n\in\bR^{n\times n}$. $\0$ (resp. $\1$) is the vector/matrix with only $0$ (resp. $1$) elements. For $x_1,\dots,x_N\in\mathbb{R}^n$, the collective vector is denoted  by $\boldsymbol{x}:=\mathrm{col}((x_i)_{i\in(1,\dots,N)})\coloneqq[x_1^\top,\dots ,x_N^\top ]^\top$. 
$\diag((A_i)_{i\in(1,\dots,N)})$ describes a block-diagonal matrix with the matrices $A_1,\dots,A_N$ on the main diagonal. The null space of a matrix $A$ is denoted as $\mathrm{ker}(A)$.
The Cartesian product of the sets $\Omega_i$, $i=1,\dots ,N$ is $\prod^N_{i=1} \Omega_i$.

\subsection{Operator-theoretic notation}
The identity operator is denoted by~$\Id(\cdot)$. 
The set valued mapping $N_{\ca C}:\bR^n\rightrightarrows \bR^n$ denotes the normal cone to the set $\mathcal{C}\subseteq \bR^n$, that is $N_{\ca C}(x)= \{ u\in\bR^n \,|\, \mathrm{sup}\langle \ca C-x,u \rangle\leq 0\}$ if $x \in \ca C$ and  $\varnothing$ otherwise. The graph of a set valued mapping $\ca A:\ca X\rightrightarrows \ca Y$ is $\mathrm{gra}(\ca A):= \{ (x,u)\in \ca X\times \ca Y\, |\, u\in\ca A (x)  \}$. For a function $\phi:\bR^n\rightarrow\overline{\mathbb{R}}$, define $\mathrm{dom}(\phi):=\{x\in\bR^n|\phi(x)<+\infty\}$ and its subdifferential set-valued mapping, $\partial \phi:\mathrm{dom}(\phi)\rightrightarrows\bR^n$, $\partial \phi(x):=\{ u\in \bR^n | \: \langle y-x|u\rangle+\phi(x)\leq \phi(y)\, , \: \forall y\in\mathrm{dom}(\phi)\}$.  The projection operator over a closed set $S\subseteq \bR^n$ is $\textrm{proj}_S(x):\bR^n\rightarrow S$ and it is defined as $\textrm{proj}_S(x):=\mathrm{argmin}_{y\in S}\lVert y - x \rVert^2$. A set valued mapping $\ca F:\bR^n\rightrightarrows \bR^n$ is $\ell$-Lipschitz continuous with $\ell>0$, if $\lVert u-v \rVert \leq \ell \lVert x-y \rVert$ for all $(x,u)\, ,\,(y,v)\in\mathrm{gra}(\ca F)$; $\ca F$ is (strictly) monotone if for all $(x,u),(y,v)\in\mathrm{gra}(\ca F)$ $\langle u-v,x-y\rangle \geq (>)0$ holds true, and  maximally monotone if it does not exist a monotone operator with a graph that strictly contains $\mathrm{gra}(\ca F)$. Moreover, it is $\alpha$-strongly monotone if for all $(x,u),(y,v)\in\mathrm{gra}(\ca F)$ it holds $\langle x-y, u-v\rangle \geq \alpha \lVert x-y \rVert^2$. The operator $\ca F$ is $\eta$-averaged ($\eta$-AVG) with $\eta\in(0,1)$ if $\lVert \ca F(x)-\ca F(y) \rVert^2 \leq \lVert x-y\rVert^2-\frac{1-\eta}{\eta}\lVert (\Id-\ca F)(x)-(\Id-\ca F)(y)  \rVert^2$ for all $x,y\in\bR^n$; $\ca F$ is $\beta$-cocoercive if $\beta\ca F$ is $\frac{1}{2}$-averaged, i.e. firmly nonexpansive (FNE).
The resolvent of an operator $\ca A:\bR^n\rightrightarrows \bR^n$ is $\mathrm{J}_{\ca A} :=(\Id+\ca A)^{-1}$.

\section{Problem Formulation}
\label{sec:problem_formulation}
\subsection{Mathematical formulation}
\label{subsec:game_formulation}

We consider a noncooperative game $\Gamma$ between $N$ agents (or players) subject to affine coupling constraints. We define the game as the triplet $\Gamma\coloneqq (\bld{\ca X}, \{f_i\}_{i\in\{1\dots N\}}, \ca G )$, where its elements are respectively: the collective feasible decision set, the players' local cost functions and the graph describing the communication network. In the following subsections, each one of them is introduced.

\subsubsection{Feasible strategy set}     
\label{subsubsec:feasibel set}  
Every agent $i\in\ca N:=\{1,\dots,N\}$ has a local decision variable (or strategy) $x_i$ belonging to its private decision set $\Omega_i\subset \bR^{n_i}$, namely the set of all those strategies that satisfy the local constraints of player $i$. The collective vector of all the strategies, or strategy profile of the game, is denoted as $\bld x\coloneqq\mathrm{col}( x_1,\dots,x_N)\in\bR^{n}$, where $n\coloneqq\sum_{i\in\ca N} n_i$. Then,  all the decision variables of all the players other than $i$ are represented via the compact notation $\bld x_{-i}\coloneqq\mathrm{col}(x_1,\dots,x_{i-1},x_{i+1},\dots,x_N)$.
We assume that the agents are subject to $m$ affine coupling constraints described by the affine function $\bld x \mapsto A\bld x +b$, where $A\in\bR^{m\times n}$ and $b\in\bR^m$.
Thus, the collective feasible decision set can be written as
 \begin{equation}
\label{eq:collective_feasible_dec_set}
\bld{\ca X} := \bld \Omega \cap \left\{\bld x\in\bR^{n} \,|\, A\bld x\leq b \right\}\, , 
\end{equation}      
where $\bld \Omega=\prod_{i\in\ca N} \Omega_i \subset \bR^{n}$, is the Cartesian product of the local constraints sets $\Omega_i$'s. Accordingly, the set of all the feasible strategies of each agent $i\in\ca N$ reads as
 \begin{equation*}
\label{eq:feasible_dec_set}
\ca{X}_i(\bld{x}_{-i}) := \left\{ y\in\Omega_i \;|\; A_i y -b_i  \leq   \textstyle{\sum_{j\in \ca N\setminus\{ i\} }} \left( b_j -A_jx_j \right) \right\}\, , 
\end{equation*} 
where $A = [A_1,\dots,A_N]$, $A_i\in\bR^{m\times n_i}$ and $\sum_{j=1}^N b_j =b$. 
The choice of affine coupling constraints is widely spread in the literature of noncooperative games, see e.g.,  \cite{yi:pavel:2019:operator_splitting_approach,Paccagnan_Gentile2016:Distributed_computation_GNE,cenedese_et_al:2019:TAC:proximal_point}. Moreover, in \cite[Remark~3]{grammatico:18tcns}, it is highlighted that separable and convex coupling constraints can always be rewritten in an  affine form.
Finally, we introduce some standard assumptions \cite{yi:pavel:2019:operator_splitting_approach,cenedese_et_al:2019:TAC:proximal_point} on the sets just introduced.
\begin{stassumption}[Convex constraint sets]
\label{ass:convex_constr_set}
For each player $i\in\ca N$, the set $\Omega_i$ is convex, nonempty and compact. The feasible local set $\ca X_i(\bld x_{-i})$ satisfies Slater's constraint qualification. 
\hfill \QEDopen
\end{stassumption}
 
 \subsubsection{Cost functions}
\label{subsubsec:cost_fn}
Each player $i\in\ca N$ has a local cost function  $f_i(x_i,\bld x_{-i}):\Omega_i\times \bld \Omega_{-i}\rightarrow \bR$, where  $\bld \Omega_{-i}\coloneqq \prod_{j\in\ca N\setminus\{i\}} \Omega_j$. The coupling between the players appears not only in the constraints but also in the cost function, due to the dependency on both $x_i$ and $\bld{x}_{-i}$. Next, we assume some properties for these functions that are extensively used in the literature \cite{facchinei:2007:GNE_problems,yi:pavel:2019:operator_splitting_approach}.   
\begin{stassumption}[Convex and differentiable cost functions]
\label{ass:convex_diff_function}
For all $i\in\ca N$, the cost function $f_i$ is continuous, continuously differentiable and convex in its first argument. 
\hfill \QEDopen
\end{stassumption}

\subsubsection{Communication network}
\label{subsubsec:communication}
The communication between agents is described by an \textit{undirected and connected} graph $\cal G =(\ca N,\ca E)$ where $\ca E \subseteq \ca N \times\ca N$ is the set of edges.  Given two agents $i,j\in\ca N$, the couple $(i,j)$ belongs to $\ca E$, if agent $i$ shares information with agent $j$ and vice versa. Then we say that $j$ is a neighbour of $i$, i.e., $j\in\ca N_i$ where $\ca N_i$ is the neighbourhood of $i$.
The number of edges in the graph is denoted by $E\coloneqq |\ca E |$. To define the \textit{incidence matrix} $V\in\bR^{E\times N}$ associated to $\ca G$, let us label the edges as $e_l$, for $l\in\{1,\dots,E\}$. We define the entry $[V]_{li}\coloneqq 1$ (resp. $-1$) if $e_l=(i,\cdot)$ (resp. $e_l=(\cdot,i)$) and $0$ otherwise. The decision of which of the two agents composing an edge is the sink and which the source is arbitrary. By construction, $V\bld 1_N = \bld 0_N$. Then, we define   $\ca E_i^{\mathrm{out}}$ (resp. $\ca E_i^{\mathrm{in}}$) as the set of all the indexes $l$ of the edges $e_l$ that start from (resp. end in) node $i$, and hence $\ca E_i=\ca E_i^{\mathrm{out}} \cup \ca E_i^{\mathrm{in}}$. 
The \textit{node Laplacian }$L\in\bR^{N\times N}$ of an undirected graph is a symmetric matrix defined by $L\coloneqq V^\top V$,  \cite[Lem.~8.3.2]{Godsil:algebraic_graph_theory}. Another important property of $L$,  used in the remainder, is $ L\bld 1_N = \bld 0_N $.

\subsection{Generalized Nash Equilibrium}
\label{subsec:qeuilibrium_concept}

In summary, the considered generalized game is described by the following set of inter-dependent optimization problems:
\smallskip
\begin{equation}\label{eq:game_formulation}
\forall i\in\ca N \::\:
\begin{cases}
 \argmin_{y\in\bR^{n_i}} &f_i(y,\bld x_{-i})\\
\qquad \text{s.t.} \; &y\in \ca{X}_i(\bld{x}_{-i})\:.
\end{cases}
\end{equation} 
\smallskip

The most popular equilibrium concept considered for noncooperative games with coupling constraints is the \textit{generalized Nash equilibrium}, thus the configuration in which all the relations in~\eqref{eq:game_formulation} simultaneously hold.

\begin{definition}[Generalized Nash Equilibrium]
\label{def:GNE}   
A collective strategy $\bld x^*\in\bld{\ca X}$ is a generalized Nash equilibrium (GNE) if, for each player $i$, it holds 
\begin{equation*}
f_i(x_i^*,\bld x_{-i}^*)\leq \mathrm{inf}\big\{\:f_i(y,\bld x_{-i}^*)\:|\: y\in \ca{X}_i(\bld{x}_{-i}^*)  \:\big\}\:.
\end{equation*}
 \vspace{-1cm}
 
 \hfill\QEDopen
\end{definition}

In this work, we focus on a subset of GNE, the so called \textit{variational} GNE (v-GNE), that has attained growing interest in the recent years-- see \cite{ facchinei:2007:GNE_problems,belgioioso:grammatico:17cdc,kulkarni:shanbhag:12}.  The name of these equilibria derives from the fact that they can be formulated as the solutions to a \textit{variational inequality }(VI). An important property of these equilibria is that each agent faces the same penalty to fulfill the coupling constraints, which is particularly useful to represent a ``fair'' competition between agents~\cite{facchinei:2007:GNE_problems}. Variational GNE can be seen as a particular case of the concept of \textit{normalized equilibrium points}, firstly introduced by Rosen in~\cite{rosen:65} and further studied in~\cite{cenedese_et_al:2019:TAC:proximal_point,paccagnan:2019:nash_ward_eq_in_aggregative_games}.

To properly characterize this set, we define the \textit{pseudo-gradient} mapping (or game mapping) of \eqref{eq:game_formulation}, as 
\begin{equation}
\label{eq:pseudo_grad}
F(\bld x)=\col\left( (\nabla_{x_i}f_i(x_i,\bld x_{-i}))_{i\in\ca N}\right)\,. 
\end{equation}
The pseudo-gradient gathers in a collective vector form the gradients  of the cost functions each w.r.t. the local decision variable. 
Next, we introduce some standard technical assumptions, e.g.,~\cite{Facchinei:2011:KKT_and_GNE,BELGIOIOSO:2017:convexity_and_monotonicity_aggr_games}.
\smallskip
\begin{stassumption}
\label{ass:subgrad_strong_mon}
The mapping $F$ in \eqref{eq:pseudo_grad}  is $\alpha$-strongly monotone and $\ell$-Lipschitz continuous, for some $\alpha, \ell >0$. \hfill\QEDopen
\end{stassumption}

When Standing assumption~\ref{ass:subgrad_strong_mon} holds true, the mapping $F$ is single valued and the set of v-GNE of the game in \eqref{eq:game_formulation} corresponds to the solution to VI($F,\bld{\ca X}$), namely the problem of finding a vector $\bld x^*\in\bld{\ca X}$ such that
\smallskip
\begin{equation}
\label{eq:VI_GNE}
 \langle F(\bld x^*),\bld x-\bld x^*\rangle \geq 0\,, \quad \forall \bld x \in \bld{\ca X} \:.
\end{equation}
The continuity of $F$ (Assumption~\ref{ass:convex_diff_function}) and compactness of $\bld{\ca X}$ (Assumption~\ref{ass:convex_constr_set}) imply the existence of a solution to VI($F,\bld{\ca X}$), while the strong monotonicity (Assumption~\ref{ass:subgrad_strong_mon}) entails uniqueness, see~\cite[Th.~2.3.3]{facchinei:pang}.      
    
Next, let us define the KKT conditions associated to the game in \eqref{eq:game_formulation}. The strong duality of the problem (Assumptions~\ref{ass:convex_constr_set}, \ref{ass:convex_diff_function}) implies that, if $\bld x^*$ is a GNE of \eqref{eq:game_formulation}, then there exist $N$ dual variables $\lambda^*_i\in\bR^m_{\geq 0}$, for all $i\in \ca N$, such that the following inclusions are satisfied:
\begin{equation} \label{eq:KKT_game}
\forall i\in\ca N :
\begin{cases}
\bld 0 \in \nabla_{x_i}f_i(x_i^*)+A_i^\top\lambda_i^*+N_{\Omega_i}(x_i^*)\, , \; \\
\bld 0 \in b-A\bld x^*+ N_{\bR^m_{\geq 0}}( \lambda^*_i)\:.
\end{cases} 
\end{equation}
Instead of looking for the solution of the general case where $\lambda_1^*,\dots,\lambda_N^*$ may be different, we examine the special case when $\lambda^*\coloneqq\lambda_1^* = \dots =\lambda_N^*$, namely
\begin{equation} \label{eq:KKT_VI}
\forall i\in\ca N :
\begin{cases}
\bld 0 \in \nabla_{x_i}f_i(x_i^*)+A_i^\top\lambda^*+N_{\Omega_i}(x_i^*)\, , \; \\
\bld 0 \in b-A\bld x^* + N_{\bR^m_{\geq 0}}( \lambda^*) \:.
\end{cases} 
\end{equation}
It follows from~\cite[Th.~3.1(ii)]{facchinei:fischer:piccialli:07}, that the KKT inclusions in \eqref{eq:KKT_VI} correspond to the solution set to VI($F, \bld{\ca X}$). Thus, every solution $\bld x^*$ to VI($F,\bld{\ca X}$) is also a GNE of the game in \eqref{eq:game_formulation}, \cite[Th.~3.1(i)]{facchinei:fischer:piccialli:07}. Since the solution set to VI($F,\bld{\ca X}$) is a singleton, we conclude that there exists a unique v-GNE of the game~\eqref{eq:game_formulation}.

\section{Synchronous Distributed GNE Seeking Algorithm}
\label{sec:synch_case}

We first introduce the synchronous counterpart of AD-GENO. The derivation of the \textit{\underline{S}ynchronous \underline{D}istributed \underline{G}N\underline{E} Seeking Algorithm with \underline{No}de variables} (SD-GENO) has as cornerstone an operator splitting approach to solve the KKT system in~\eqref{eq:KKT_VI}. 
Originally proposed in \cite{yi:pavel:2019:operator_splitting_approach,belgioioso:grammatico:17cdc} in the contest of GNE finding problems.

\subsection{Algorithm design}
\label{sec:alg_develop_synch}


The KKT conditions of each agent $i$ in~\eqref{eq:KKT_game} are satisfied by a couple $(x_i,\lambda_i)$, where the dual variables $\lambda_i$ may be different among the players. If we enforce the consensus among the dual variables, then the unique solution of the inclusions is the v-GNE of the game. This is achieved by exploiting the fact that $\ker(V) = \mathrm{span}(\1)$ and introducing  the auxiliary variables $\sigma_l,\, l\in\{1,\dots,E\}$, one for every edge in the graph.
Using the notations $\bld \lambda \coloneqq \col((\lambda_i)_{i\in\ca N})\in\bR^{mN}$, $\Lambda \coloneqq \diag((A_i)_{i\in\ca N})\in\bR^{mN\times n}$, $\bar b \coloneqq \col ((b_i)_{i\in\ca N})\in\bR^{mN}$, $\bld \sigma\coloneqq \col ((\sigma_l)_{l\in\{1\dots  E\}})\in\bR^{mE}$, $\bld V \coloneqq V\otimes I_m \in \bR^{mE\times mN}$ and $\bld L \coloneqq L\otimes I_m \in \bR^{mE\times mN}$, we cast the augmented version of the inclusions in \eqref{eq:KKT_game} by
\smallskip
\begin{equation} \label{eq:mod_compact_KKT_game}
\begin{split}
\bld 0 &\in F(\bld x)+\Lambda^\top \bld \lambda+N_{\bld \Omega}(\bld x) \\
\bld 0 &\in \bar b-\Lambda\bld x + N_{\bR^{mN}_{\geq 0}}(\bld \lambda)+ \bld{L\lambda} +\rho  \bld V^\top \bld \sigma \\
\bld{0} &= -\rho \bld{  V  \lambda} \:,
\end{split} 
\end{equation}
where $\rho\in\bR_{>0}$.

A solution $\varpi = \col (\bld x^*,\bld \sigma^*,\bld \lambda^*)$ of the above inclusions can be equivalently recast as a zero of the sum of two mappings $\ca A$ and $\ca B$ defined as 
\smallskip
\begin{equation}
\label{eq:operators_def}
\begin{split}
\ca A : \varpi \mapsto &\begin{bmatrix}
0 & 0 & \Lambda^\top  \\
0 & 0 & -\rho\bld V  \\
 -\Lambda & \rho\bld V^\top & 0
\end{bmatrix}\varpi +\begin{bmatrix}
N_{\bld \Omega} (\bld x)\\
0\\
N_{\bR^{mN}_{\geq 0}} (\bld \lambda)
\end{bmatrix}\\
\ca B : \varpi \mapsto &\begin{bmatrix}
F(\bld x)\\
0\\
\bar b + \bld{L\lambda}
\end{bmatrix}\,.
\end{split}
\end{equation} 
In fact, $\varpi^*\in\mathrm{zer}(\ca A+\ca B)$ if and only if $\varpi^*$ satisfies \eqref{eq:mod_compact_KKT_game}.

Next, we show that the zeros of $\ca A+\ca B$ characterize the v-GNE of the original game. 
\begin{proposition}\label{prop:zer_AB_are_vGNE}
Let $\ca A$ and $\ca B$ be as in \eqref{eq:operators_def}. Then the following hold:
\begin{enumerate}[(i)]
\item $\mathrm{zer}(\ca A +\ca B)\not = \varnothing$ ,
\item if $ \col (\bld x^*,\bld \sigma^*,\bld{ \lambda}^*)\in\mathrm{zer}(\ca A +\ca B)$ then $(\bld x^*,\bld \lambda^*)$ satisfies the KKT conditions in \eqref{eq:KKT_game}, with $\lambda_1^*=\dots=\lambda_N^*$, hence $\bld x^*$ is the unique v-GNE of the game in \eqref{eq:game_formulation}.

 \hfill\QEDopen 
\end{enumerate} 
\end{proposition}
 The proof is attained by exploiting the property that $\mathrm{ker}(V)=\mathrm{ker}(L)$, for the graph described in Section~\ref{subsubsec:communication}. The steps are similar to those in \cite[Th.~2]{yi:pavel:2019:operator_splitting_approach}.  We omit them here for brevity reasons.

Several researchers have analyzed the problem of finding a zero of the sum of two monotone operators. The so called \textit{splitting methods} represent one of the most popular approach developed to attain an iterative algorithm to solve this class of problem - see \cite{eckstein1989:splitting}, \cite[Ch.~26]{bauschke2011convex}.  

\smallskip
\begin{lemma}\label{lem:max_mon_of_operators}
The mappings $\ca A$ and $\ca B$ in \eqref{eq:operators_def} are maximally monotone. Moreover, $\ca B $ is  $\chi $-cocoercive, where $\chi \coloneqq \min\left\{\tfrac{\alpha}{\ell^2},\,\lambda_{\max}(L)^{-1} \right\}$. \hfill\QEDopen
\end{lemma}
\begin{proof}
See \ref{app:proof_synch_case}
\end{proof}

The properties of the operators proved above drive us to select the  \textit{preconditioned forward-backward} splitting (PFB) to derive a distributed and iterative algorithm seeking $\mathrm{zer}(\ca A+\ca B)$. 
This approach was previously adopted by other researchers, e.g., \cite{yi:pavel:2019:operator_splitting_approach}.

The PFB splitting operator reads as 
\begin{equation}
\label{eq:T_PFB_operator}
T \coloneqq \mathrm{J}_{\Phi^{-1}\ca A}\circ(\Id-\Phi^{-1}\ca B)\, .
\end{equation}
The so-called preconditioning matrix $\Phi$ is defined by
 \begin{equation}
 \label{eq:preconditioning_matrix}
 \Phi:=\begin{bmatrix}
	\bld \tau^{-1}   & 0  & -\Lambda^\top\\
	0 & \delta^{-1}I_{mM}  & \rho\bld V \\
	-\Lambda & \rho\bld V^\top & \bld \varepsilon^{-1}
 \end{bmatrix}
\end{equation}   
where $\delta\in\bR_{>0}$,  $\bld \varepsilon = \diag((\varepsilon_i)_{i\in\ca N})\otimes I_m $ with $\varepsilon_i>0$, for all $i\in \ca N$ and $\bld \tau$ is defined in a similar way.

The update rule of the algorithm is obtained by including a relaxation step, i.e.,
\begin{equation}
\label{eq:Krasno_iter_synch}
\begin{split}
\tilde\varpi(k) &= T \varpi(k) \\
\varpi(k+1)&=\varpi(k)+\eta (\tilde\varpi(k)-\varpi(k))\:.
\end{split}
\end{equation} 

It comes from \eqref{eq:T_PFB_operator} that $\fix(T)=\mathrm{zer}(\ca A + \ca B)$, in fact $\varpi \in\fix(T) \iffs \varpi \in T \varpi \iffs 0 \in \Phi^{-1}(\ca A+\ca B)\varpi \iffs \varpi\in\zer(\ca A +\ca B)$, see~\cite[Th.~26.14]{bauschke2011convex}. 

In the remainder of this section, we provide the complete derivation of SD-GENO, obtained directly from~\eqref{eq:Krasno_iter_synch}. In the following, we denote  $\varpi\coloneqq\varpi(k)$, $\varpi^+\coloneqq\varpi(k+1)$ and $\tilde\varpi\coloneqq \tilde\varpi(k)$ to simplify the notation. 
Consider  $\tilde \varpi=T\varpi$. From \eqref{eq:T_PFB_operator} it holds that $\tilde \varpi = \mathrm{J}_{\gamma\Phi^{-1}\ca A}\circ(\Id-\gamma\Phi^{-1}\ca B)\varpi \iffs \Phi(\varpi -\tilde\varpi) \in \ca A \tilde \varpi +\ca B \varpi$, thus
\begin{equation}\label{eq:inclusion_synch}
\bld 0\in  \ca A \tilde \varpi +\ca B \varpi + \Phi(\tilde\varpi -\varpi)\:.
\end{equation}
The update rule of each components of $\varpi$ is attained by analyzing the row blocks of~\eqref{eq:inclusion_synch}.
The first reads as
 $\bld 0 \in F(\bld x) +N_{\bld \Omega}(\tilde{\bld x}) + \bld \tau^{-1}(\tilde{\bld x} -\bld x) + \Lambda^\top\bld \lambda$. By solving this inclusion by $\bld{\tilde x}$, we attain the update rule for the primal variables:
\begin{equation}\label{eq:row1_sync}
\tilde{\bld x} = \mathrm{J}_{N_{\bld \Omega}} \circ\big(\bld x-\bld \tau(F(\bld x)+\Lambda^\top \bld \lambda ) \big)\,.
\end{equation}
Similarly, from the second row block of~\eqref{eq:inclusion_synch}, we attain the update for $\bld{\tilde\sigma}$, i.e., 
\smallskip
\begin{equation}\label{eq:row2_sync}
\tilde{\bld \sigma} = \bld \sigma +  \delta \rho\bld V\bld \lambda \:.
\end{equation} 
Finally, the third row block of \eqref{eq:inclusion_synch} is $\bld 0 \in \bar b +\bld{L\lambda}+N_{\bR^{mN}_{\geq 0}}(\tilde{\bld \lambda}) +\Lambda(2\tilde{\bld x}-\bld x) + \rho\bld V^\top(2\tilde{\bld \sigma}-\bld \sigma)+ \bld \varepsilon^{-1}(\tilde{\bld\lambda}-\bld\lambda)$, from which we obtain 
\begin{equation}\label{eq:row3_sync}
\begin{split}
\tilde{\bld \lambda} &= \mathrm{J}_{N_{\bR^{mN}_{\geq 0}}}\circ \big(\bld \lambda+\bld \varepsilon( \Lambda (2\tilde{\bld x}-\bld x) -\bar b - \rho\bld V^\top(2\tilde{\bld \sigma}-\bld\sigma ) -\bld{L\lambda}) \big) \\
& \stackrel{\text{\eqref{eq:row2_sync}}}{=}  \mathrm{proj}_{\bR^{mN}_{\geq 0}}\big(\bld \lambda+\bld \varepsilon( \Lambda (2\tilde{\bld x}-\bld x) -\bar b \\
& \hspace{3cm} - \rho\bld V^\top\bld \sigma -(2\delta\rho^2+1)\bld{L\lambda}) \big)\:. \\
\end{split}
\end{equation}

Note that, the update of $\bld{\tilde \lambda}$ depends only on the aggregate information $\bld V^\top \bld \sigma$. We can exploit this feature to replace the edge auxiliary variables $\sigma_l$'s, with a single variable for each agent $i$ defined by  $z_i\coloneqq \big([V^\top]_i\otimes I_m\big) \,\bld{\sigma}\in\bR^{mN}$. Recalling that $\bld V^\top\bld V =L\otimes I_m\eqqcolon \bld L$, we compute the update rule of these new variables and replace \eqref{eq:row2_sync} by
\begin{equation}\label{eq:row2_sync_z}
\tilde{\bld z} =\bld V^\top \bld{\sigma} +  \delta \rho\bld V^\top\bld V\bld \lambda  =\bld z +  \delta \rho\bld L\bld \lambda \:.
\end{equation} 
Consequently, \eqref{eq:row3_sync} is modified accordingly as
  \begin{equation}\label{eq:row3_sync_z}
\begin{split}
\tilde{\bld \lambda} &= \mathrm{proj}_{\bR^{mN}_{\geq 0}} \big(\bld \lambda+\bld \varepsilon( \Lambda (2\tilde{\bld x}-\bld x) -\bar b \\
& \hspace{3cm} - \rho\bld{z} -(2\delta\rho^2+1)\bld{L\lambda}) \big)\:. \\
\end{split}
\end{equation} 
To ensure that this change of variables does not affect the equilibrium of the game, we introduce the following result proving that an equilibrium point of the new set of equations is indeed a v-GNE of~\eqref{eq:game_formulation}.
\smallskip
\begin{theorem}\label{th:eq_are_vGNE_mod_map}
If $\col(\bld x^*, \bld z^*, \bld \lambda^* )$ is a solution to the equations \eqref{eq:row1_sync}, \eqref{eq:row2_sync_z},  \eqref{eq:row3_sync_z}, with $\bld 1^\top\bld z^*=0$, then $\bld x^*$ is a v-GNE of the game in~\eqref{eq:game_formulation}.   
\hfill\QEDopen 
\end{theorem}
\begin{proof}
See \ref{app:proof_synch_case}.
\end{proof}
\smallskip
\begin{remark}
\label{rem:O(1)}
In~\cite{yi:pavel:2019:asynch_distributed_GNE_w_partial_info}, the algorithm  SYDNEY  achieves convergence to the v-GNE of the game~\eqref{eq:game_formulation}, when this is subject to equality coupling constraints only. This solution relies on edge auxiliary variables to enforce the consensus of the $\lambda_i$'s. Therefore, the number of variables that each agent has to store is $\ca O(N)$. 

The change of ``variables'', from $\bld \sigma$ to $\bld z$, is particularly useful in large not-so-sparse networks and it is in general convenient when the  edges outnumber the nodes. In fact, it implies that the number of variables grows linearly with the number of nodes, and thus the memory requirement for each player is $\ca O(1)$. This fact and the possibility to handle affine coupling constraints represent the main advantages of adopting SD-GENO over SYDNEY.  \hfill\QEDopen 
\end{remark}

\subsection{Synchronous, distributed algorithm with node variables (SD-GENO)} 
\begin{algorithm}[t]
\DontPrintSemicolon
\textbf{Input:} $k=0$, for all $i\in\ca N$, $x_i(0) \in \bR^{n_i} $, $\lambda_i(0) \in\bR^{m}$, $ z_i(0)=\bld 0_{m}$.\;
\hspace{0.1cm} Choose $\delta,\, \varepsilon_i,\, \tau_i$ satisfying~\eqref{eq:step_choices}, while $\eta\in(0,1)$ and $\rho\in(0,1]$. \;
 \textbf{Iteration $k$:} \\
 \textbf{Communication:} each $i\in\ca N$ gathers $\lambda_j(k)$ from the neighbors and updates the disagreement vector $d_i(k)\coloneqq \sum_{j\in\ca N_i }  (\lambda_{i} - \lambda_{j})$ \;
 \textbf{Local update, }
\For{$i\in\ca N$}{
$\tilde{ x}_{i}=\mathrm{proj}_{\Omega_i} \big( x_{i}-\tau_i(\nabla_i f_i(x_{i},\bld x_{-i})+ A_i^\top \lambda_{i} ) \big) $ \;
$\tilde{z}_{i} = z_{i} +\rho\,\delta\, d_i(k) $\;
$\tilde{\lambda}_{i}  = \mathrm{proj}_{\bR^{m}_{\geq 0}} \big( \lambda_{i}+\varepsilon_i\left( A_i(2\tilde{x}_{i} - x_{i}) - b_i\right. $\;
$\hspace{2cm } -\rho z_{i} -(2\delta\rho^2+1)d_i(k) \big) \big)$\;
$x_{i}^+ =  x_{i} +\eta(\tilde{ x}_{i} - x_{i})$\;
$z_{i}^+ =  z_{i} +\eta(\tilde{ z}_{i} -z_{i})$ \;
$\lambda_{i}^+ = \lambda_{i} +\eta(\tilde{ \lambda}_{i} - \lambda_{i})$\;
}
$k\leftarrow k+1$\;
\caption{SD-GENO}
\label{alg:synch_alg}
\end{algorithm}

The complete formulation of the algorithm is obtained by gathering together all the update rules introduced in the previous section, i.e.,  \eqref{eq:row1_sync}, \eqref{eq:row2_sync_z}, \eqref{eq:row3_sync_z} and adding a relaxation step. The algorithm in compact form is expressed as 
\begin{equation}\label{eq:update_compact_complete_alg}
\begin{cases}
&\tilde{\bld x} = \mathrm{proj}_{\bld \Omega} \big(\bld x-\bld \tau(F(\bld x)+\Lambda^\top \bld \lambda ) \big) \\
&\tilde{\bld z} = \bld z + \rho \delta \bld L \bld \lambda\\
&\tilde{\bld \lambda} = \mathrm{proj}_{\bR^{mN}_{\geq 0}} \big(\bld \lambda+\bld \varepsilon( \Lambda (2\tilde{\bld x}-\bld x) -\bar b\\
&\hspace{3cm} -\rho\bld z -(2\delta\rho^2+1)\bld{L\lambda}) \big)\\
&\bld x^+ = \bld x +\eta(\tilde{\bld x} -\bld x)\\
&\bld z^+ = \bld z  +\eta(\tilde{\bld z} -\bld z)\\
&\bld \lambda^+ = \bld \lambda +\eta(\tilde{\bld \lambda} -\bld \lambda) \:,\\
\end{cases}
\end{equation}  
while the local updates and the initial condition of SD-GENO are provided in Algorithm~\ref{alg:synch_alg}.
%
It is composed of two main phases: the communication with the neighbors and the local update. First, each agent gathers the information about the strategies and the dual variables of the neighbors. Next, the local update is performed, based on a gradient descend and dual ascend structure. It is worth noticing that only one round of communication is required at each iteration of SD-GENO.

The convergence of SD-GENO to the v-GNE of the game in \eqref{eq:game_formulation} is proven in the following theorem.
\smallskip
\begin{theorem}\label{th:convergence_sync}
Set the step sizes $\varepsilon_i$, $ \delta$, $\tau_i$, for all $i\in\ca N$, such that
\begin{subequations}
\label{eq:step_choices}
\begin{equation}
\tau_i  \leq (\lVert A_i\rVert+\vartheta)^{-1}
\end{equation}
\begin{equation}
\label{eq:delta}
\delta  \leq (2 \rho+\vartheta)^{-1} 
\end{equation}
\begin{equation}
\label{eq:epsilon}
\varepsilon_i  \leq (\rho|\ca N_i| +\lVert A_i\rVert  +\vartheta )^{-1} \:,
\end{equation} 
\begin{equation}
\vartheta>\frac{1}{2\chi}
\end{equation}
\end{subequations}
 with $\chi$ as in Lemma~\ref{lem:max_mon_of_operators} and $\eta\in\big(0,\frac{4\chi\vartheta-1}{2\chi\vartheta}\big)$.
Then, the sequence $(\bld x(k))_{k\in\bN}$ generated by SD-GENO (Algorithm~\ref{alg:synch_alg}) converges to the v-GNE of the game in \eqref{eq:game_formulation}.    
\hfill\QEDopen
\end{theorem}
\begin{proof}
See \ref{app:proof_synch_case}.
\end{proof}
\smallskip
 
\section{Asynchronous Distributed Algorithm with Edge Variables (AD-GEED)}
\label{sec:asynch_case}
In the  case of heterogeneous agents with very different update rates, SD-GENO can converge slowly, due to its synchronous structure. To overcome this limitation, we introduce here the   \textit{\underline{A}synchronous \underline{D}istributed \underline{G}N\underline{E} Seeking Algorithm with \underline{Ed}ge variables} (AD-GEED). It uses edge auxiliary variables $\{\sigma_l\}_{l\in \{1\dots E\}}$ and an asynchronous update to compute the v-GNE of the game in~\eqref{eq:game_formulation}. As discussed in the previous section, this is not an ideal solution due to the poor scalability in the case of dense networks. A  refinement of this algorithm, which relies on node variables only, is provided in Section~\ref{sec:AD_GENO}. 
From a technical point of view, the asynchronicity is achieved by exploiting an asynchronous framework for fixed-point iterations, the  so called  ``ARock'' framework, developed in~\cite{peng2016arock}.
\smallskip
\subsection{Algorithm design}
The update rule in the asynchronous case, is similar to that in~\eqref{eq:Krasno_iter_synch}, with the main difference that, at each time instant, only one agent $i\in\ca N$  updates its strategy $x_i$, dual variable $\lambda_i$ and local auxiliary variables $\{\sigma_l\}_{l\in\ca E_i^{\textup{out}}}$. 
To mathematically formulate this concept we introduce $N$ diagonal matrices $\bld H_i$, where $[\bld H_{i}]_{jj}$ is $1$ if the $j$-th element of $\col( \bld x,\: \bld\sigma,\:\bld \lambda)$ is an element of $\col(x_i,\:\{\sigma_l\}_{l\in\ca E ^{\mathrm{out}}_i},\: \lambda_i)$ and $0$ otherwise. The matrix $\bld H_i$ triggers the update of those elements in $\varpi$ that are associated to agent $i$. 
We assume that the choice of the agent performing the update at time instant $k$ is ruled by an i.i.d. random  variable $\zeta(k)$, taking values in $\bld H:=\{ \bld H_i \}_{i\in\ca N}$. Given a discrete probability distribution $(p_1,\dots,p_N)$, let $\mathbb{P}[\zeta(k) = \bld H_i] = p_i$, for all $i\in \ca N$. Therefore, the update rule in the asynchronous case is cast as
\begin{equation}\label{eq:Krasno_asynch_zeta}
\varpi(k+1)=\varpi(k)+\eta \zeta(k)(T\varpi(k)-\varpi(k))\,.
\end{equation}

An illustrative example is now provided to clarify how to construct $\bld H$.  

\begin{example}
\label{ex:matrix_H}
Consider a game with $N=3$, $E=2$, $m=1$, $n_i=1,\, i=1,2,3$ and $\varpi$ is the collective vector of all the strategies and auxiliary variables in the game. The communication network is described by the undirected graph $\ca G$, where the arrows describe the convention adopted for the edges. 
\begin{figure}[H]
\centering
\includegraphics[scale = 0.25]{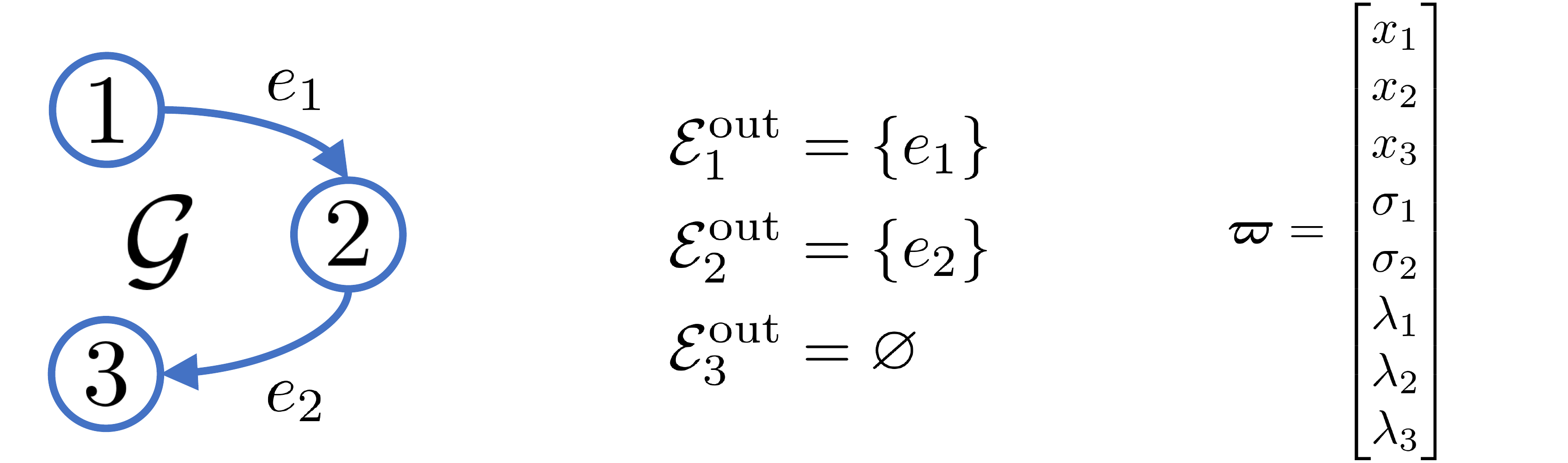}
\end{figure}
In this case, $\bld H$ is a  set of three $8\times 8$ matrices, namely
\begin{equation*}
\begin{split}
\bld H_1 &\coloneqq \diag((1,0,0,1,0,1,0,0)) \\
\bld H_2 &\coloneqq \diag((0,1,0,0,1,0,1,0)) \\
\bld H_3 &\coloneqq \diag((0,0,1,0,0,0,0,1)) \:. \\
\end{split}
\end{equation*}
If at time $k$ agent $2$ is updating, ~\eqref{eq:Krasno_asynch_zeta} turns into
\begin{equation}
\varpi(k+1)=\varpi(k)+\eta \bld H_2(T\varpi(k)-\varpi(k))\,.
\end{equation}
So, the only elements of $ \varpi$ that change are $(x_2,\sigma_2,\lambda_2)$, precisely the variables associated to agent $2$.\hfill\QEDopen
\end{example}
\smallskip

In addition to asynchronicity, we generalize~\eqref{eq:Krasno_asynch_zeta} by considering possible delays in the information used by each agent in the local update, i.e., the information gathered from the neighbors may be outdated, denoted as ${\hat \varpi}$. These delays are due to a non-neglectable computation time  for the update of the agents, refer to \cite[Sec.~1]{peng2016arock} for a more complete overview on the topic. All the variables updated by the same agent share the same delay, e.g.,  $x_i$, $\lambda_i$ and $\{\sigma_l\}_{l\in\ca E ^{\mathrm{out}}_i} $ has a delay $\varphi_i(k)\geq 0$ at time instant $k$.  

\begin{algorithm}[!t]
\DontPrintSemicolon
\textbf{Input:} $k=0$, $\bld x^0 \in \bR^{n} $, $\bld \lambda^0 \in\bR^{mN}$, $\bld \sigma^0=\bld 0_{mM}$, chose $\delta,\,  \varepsilon_i ,\, \tau_i$ satisfying~\eqref{eq:step_choices} and $\eta\in(0,1)$.  \;
\hrule
\smallskip
\textbf{Iteration $k$:} Select the agent $i_k$ with probability $\mathbb{P}[\zeta(k)=\bld H_{i_k}]=p_{i_k}$\; 
\textbf{Reading:} Agent $i_k$ copies in its private memory the current values of the public memory, i.e. $\hat x_{j} $, $\hat \lambda_{j} $, $\forall  j\in\ca N_{i_k}$ and $\hat \sigma_{l}$, $\forall l\in\ca E^{\textup{in}}_{i_k}$ and $l\in\ca E^{\textup{out}}_{j}$.\; 

\vspace{0.25cm}
\textbf{Update:}\;
$\tilde{ x}_{i_k} = \mathrm{proj}_{\Omega_{i_k}} \big( x_{i_k}-\tau_{i_k}(\nabla_{i_k} f_{i_k}(x_{i_k},\hat{\bld x}_{-i_k}) + A_{i_k}^\top \lambda_{i_k} \big) $ \;\bigskip
$\tilde{\sigma}_{l} = \sigma_{l} + \delta\rho ([V]_{l}\otimes I_m )\hat{\bld\lambda} \:,\quad \forall l\in \ca E_{i^k}^{\mathrm{out}}$\;\bigskip
$\tilde{\lambda}_{i_k}  = \mathrm{proj}_{\bR^{m}_{\geq 0}} \bigg( \lambda_{i_k}+\varepsilon_{i_k}\big( A_{i_k}(2\tilde{x}_{i_k} - x_{i_k}) - b_{i_k} -  \rho([V^\top]_{i_k}\otimes I_m)\bld{\hat \sigma}  - (2\delta\rho^2+1)\sum_{j\in\ca N_{i_k}} (\lambda_i  - \hat \lambda_{j})\big) \bigg)$\;\bigskip
$x_{i_k}^+ =  x_{i_k} +\eta(\tilde{ x}_{i_k} - x_{i_k})$\;
$\sigma_{l}^+ =  \sigma_{l} +\eta(\tilde{ \sigma}_{l} -\sigma_{l})\:,\quad \forall l\in \ca E_{i^k}^{\textup{out}} $ \;
$\lambda_{i}^+ = \lambda_{i_k} +\eta(\tilde{ \lambda}_{i_k} - \lambda_{i_k})$\;

\vspace{0.25cm}
\textbf{Writing:} in the public memories of each  $j\in\ca N_{i_k}$\; 
$(x_{i_k}, \lambda_{i_k})\leftarrow (x_{i_k}^+, \lambda_{i_k}^+)$\;
$\{\sigma_{l}\}_{l\in\ca E^{\textup{out}}_{i_k}} \leftarrow \{\sigma_{l}^+\}_{l\in\ca E^{\textup{out}}_{i_k}}$\; \smallskip
 $k\leftarrow k+1$\;
\caption{AD-GEED}
\label{alg:AD-GEED}
\end{algorithm}

According to this, the final formulation of the update rule~\eqref{eq:Krasno_asynch_zeta} becomes

\begin{equation}\label{eq:Krasno_asynch_final}
\varpi(k+1)=\varpi(k)+\eta \zeta(k)(T-\Id)\hat\varpi(k)\,.
\end{equation} 

The only assumption that we impose over the delay, is boundedness, as formalized next. 
\smallskip
\begin{assumption}[Bounded maximum delay]\label{ass:bounded_delay}
The delays are uniformly upper bounded, i.e. there exists $\bar\varphi>0$ such that $\sup_{k\in\bN_{\geq 0}}\max_{i\in\ca N}\{ \varphi_i(k) \}\leq \bar\varphi<+\infty$. 
\hfill\QEDopen
\end{assumption} 
\smallskip

We assume that each agent $i$ is equipped with a public and private memory, the first one is used by the neighbors to write their strategy (and dual/auxiliary variables) at the end of each update. At the beginning of an update, agent $i$ copies the values from the public to the private memory, and uses them to complete the update.
 Notice that, during the update time of $i$, the neighbors can still write in the public memory of $i$ without affecting the values stored in the private one.   
The local update rules of AD-GEED are presented in Algorithm~\ref{alg:AD-GEED} and they are obtained with steps similar to those introduced in Sec.~\ref{sec:alg_develop_synch} for SD-GENO. To ease the notation, for each agent $j\in\ca N$, we define  $\hat x_{j} \coloneqq x_{j}(k+\varphi_j(k))$, $\hat \lambda_{j} \coloneqq \lambda_{j}(k-\varphi_j(k))$ and $\hat\sigma_{l} \coloneqq \sigma_{l}(k-\varphi_j(k))$, for all $l\in\ca E^{\textup{out}}_{j}$, and furthermore $\bld{\hat x} \coloneqq \mathrm{col}((\hat x_j)_{j\in\ca N})$, $\bld{\hat \lambda} \coloneqq \mathrm{col}((\hat\lambda_j)_{j\in\ca N})$, 
$\bld{\hat \sigma} \coloneqq \mathrm{col}((\hat\sigma_j)_{j\in\ca N})$. Notice that each agent has always access to the most recent value of its variables, i.e., for agent $i$ the delay $\varphi_i(k)=0$ for every $k$.

The following convergence theorem  is achived by exploiting the results in \cite{peng2016arock} for a Krasnosel'ski\u i asynchronous iteration.
\smallskip
\begin{theorem}\label{th:convergence_asynch_sigma} 
For every $i\in\ca N$, choose  $\varepsilon_i$, $ \delta$, $\tau_i$ as in~\eqref{eq:step_choices}, and let  $\eta\in \left( 0,   \frac{cNp_{\min}}{2\bar\varphi\sqrt{p_{\min}}+1} \left(2- \frac{1}{2\chi\vartheta}\right)\right]$ and $c\in(0,1)$. Then, the sequence $(\bld x(k))_{k\in\bN}$ generated by AD-GEED (Algorithm~\ref{alg:AD-GEED}) converges to the v-GNE of the game in \eqref{eq:game_formulation} almost surely.  \hfill\QEDopen
\end{theorem}
\begin{proof}
See \ref{app:prrof_AD_GEED}.
\end{proof}
\begin{remark}
If the probability distribution is uniform, i.e., $p_{\min}=\frac{1}{N}$, and we choose $\vartheta = \frac{1}{\chi}$, then the bounds on the relaxation step simplify as $\eta\in \left( 0,  \frac{3}{2}\frac{c \sqrt{N}}{2\bar\varphi+\sqrt{N}}   \right]$. Moreover, if there is no delay, so $\bar \varphi = 0$, or the number of agents is very high, the bounds may be chosen independently from the number of players, e.g., as $\eta \in(0, 1]$. \hfill \QEDopen
\end{remark}

The structure of AD-GEED is similar to that of ADAGNES in~\cite[Alg.~1]{yi:pavel:2019:asynch_distributed_GNE_w_partial_info}, where edge auxiliary variables are used to achieve consensus over the dual variables. However, unlike ADAGNES, our algorithm can handle inequality coupling constraints.

\section{Asynchronous, distributed algorithm with node variables (AD-GENO)}
\label{sec:AD_GENO}

This section presents the main result of the paper, namely, we refine AD-GEED to obtain an algorithm with the same performance, in terms of convergence speed, but relying on node auxiliary variables only, i.e., the  \textit{\underline{A}synchronous \underline{D}istributed \underline{G}N\underline{E} Seeking Algorithm with \underline{No}de variables} (AD-GENO). 
Using Algorithm~\ref{alg:AD-GEED} as the starting point, we notice that the local update of $\lambda_i$ requires only the aggregate quantity $([V^\top]_i\otimes I_m)\bld \sigma$. Our key idea is to introduce a variable $z_i$ to capture the variation of this aggregate quantity and show that it does not affect the dynamics of the pair $( x_i, \lambda_i)$, thus preserving the convergence. 
Unlike the synchronous case, we cannot directly define $\bld z = \bld V^\top \bld \sigma$, due to the different update frequency between $\{\sigma_l\}_{l\in\ca E_i}$ and $z_i$ that would affect the dynamics of $\bld \lambda$. This mismatch is clarified via the following example.

\smallskip
\begin{example}\label{ex:sigma_z}
Consider the communication network in Example~\ref{ex:matrix_H} and assume that in the first three time instances, agent $2$ updates twice and then $3$ updates once, i.e., $i_0=i_1=2$ and $i_2=3$. For $k=1$, according to Algorithm~\ref{alg:AD-GEED}  it holds
\begin{equation}
\begin{split}
\sigma_2(2) &=  \sigma_2(1) + \eta\rho\delta (\lambda_2(1)-\lambda_3(0)) \\
\lambda_2(2) & \propto \rho(\sigma_2(1)-\sigma_1(0)) \:,
\end{split}
\end{equation}
where $\propto$ is used to describe dependency. Next, for $k=2$ only $\lambda_3 $ is updated, then
\begin{equation}
\label{eq:ex2_lambda}
\lambda_3(3) \propto -\rho\sigma_2(2)\:. 
\end{equation} 
If we substitute the edge variables $\sigma_1$, $\sigma_2$ with $z_i=[V^\top]_i\bld \sigma$ for $i=1,2,3$, and apply the same activation sequence, it leads to
\begin{equation}
\label{eq:ex2_zeta}
\begin{split}
z_3(3) &= z_3(0) + \eta\rho\delta (\lambda_3(0)-\lambda_2(2))\\ 
\lambda_3(3) & \propto \rho z_3(0)\:. 
\end{split}
\end{equation}
From the comparison of \eqref{eq:ex2_lambda} and \eqref{eq:ex2_zeta}, it is clear that the value of $\lambda_3(3)$ would be different in the two cases. This is explained by the fact that $\sigma_2$ is updated twice, while $z_3$ only once. \hfill\QEDopen   
\end{example}

\smallskip
\begin{algorithm}[!t]
\DontPrintSemicolon
\textbf{Input:} $k=0$, $\bld x(0) \in \bR^{n} $, $\bld \lambda(0) \in\bR^{mN}$, $\bld z(0)=\bld 0_{mN}$.  For every $i\in\ca N$, choose $\delta,\, \varepsilon_i,\, \tau_i$ satisfying~\eqref{eq:step_choices}, $\eta\in(0,1)$ and set $\mu_i=\bld 0_m$.  \;
\hrule
\smallskip
\textbf{Iteration $k$:} Select the agent $i_k$ with probability $\mathbb{P}[\zeta(k)=H_{i_k}]=p_{i_k}$\; 
\textbf{Reading:} Agent $i_k$ copies its public memory in the private one, i.e., the values $\hat x_{j}$, $\hat \lambda_{j}$, $\forall  j\in\ca N_{i_k}$, and $\mu_{i}$.\;
\hspace{0.1cm} Reset the public values of $\mu_{i}$ to  $\bld 0_m$.\; 

\vspace{0.25cm}
\textbf{Update:}\;
$\tilde{ x}_{i_k}  = \mathrm{proj}_{ \Omega_{i_k}} \big( x_{i_k}-\tau_{i_k}(\nabla_{i_k}  f_{i_k}(x_{i_k},\hat{\bld x}_{-i_k})+ A_{i_k}^\top \lambda_{i_k} ) \big) $ \;\smallskip
$\tilde{z}_{i_k} = z_{i_k} +\delta\eta\mu_{i_k}  $\;\smallskip
$\tilde{\lambda}_{i_k}  = \mathrm{proj}_{\bR^{m}_{\geq 0}} \left( \lambda_{i_k}+\varepsilon_{i_k}( A_{i_k}(2\tilde{x}_{i_k} - x_{i_k}) \right. - b_{i_k}-\rho\tilde z_{i_k}\left. +(2\delta\rho^2-1)\sum_{j\in\ca N_{i_k}\setminus \{i_k\}} (\lambda_{i_k}-\hat \lambda_{j}) \right)$\;\bigskip
$x_{i_k}^+ =  x_{i_k} +\eta(\tilde{ x}_{i_k} - x_{i_k})$\;
$z_{i_k}^+ =  \tilde z_{i_k} + \eta\delta\rho\sum_{l\in\ca E_{i_k}^{\mathrm{out}}} ([V]_l\otimes I_m )\hat{\bld \lambda} $ \;
$\lambda_{i_k}^+ = \lambda_{i_k} +\eta(\tilde{ \lambda}_{i_k} - \lambda_{i_k})$\;

\vspace{0.25cm}
\textbf{Writing:} in the public memory of each  $j\in\ca N_{i_k}$\;
$(x_{i_k},\lambda_{i_k}) \leftarrow (x_{i_k}^+,\lambda_{i_k}^+)$\;
$\mu_{j} \leftarrow \mu_{j} + \hat\lambda_{j} - \lambda_{i_k}$\;\smallskip
  
 $k\leftarrow k+1$\;
\caption{AD-GENO}
\label{alg:AD-GENO}
\end{algorithm}
\smallskip

To bridge the gap between $\bld \sigma$ and $\bld z$, we introduce an extra variable $\mu_i\in\bR^{m}$ for each node $i$. The role of $\mu_i$ is to store the changes of the neighbors dual variable $\lambda_j$ during the time between the last update of $i$ and the next one. The scalability of the algorithm is not affected by these additional variables, one for each agent. Therefore, the benefit of adopting only node variables, highlighted in Remark~\ref{rem:O(1)}, still hold also in this asynchronous counterpart. 
Furthermore, also the number of required communications rounds between agents does not increase, since the variable $\mu_i$ is updated by the neighbors of agent $i$ during their writing phase.
Algorithm~\ref{alg:AD-GENO} presents the local update rules of AD-GENO.

The convergence of AD-GENO is proven by the following theorem. Essentially, we show that introducing $\bld z$ and $\mu$ does not affect the dynamics of $(\bld x,\bld \lambda)$. 
\smallskip
\begin{theorem}\label{th:convergence_AD-GENO}
For every $i\in\ca N$ choose $\varepsilon_i$, $ \delta$, $\tau_i$ as in~\eqref{eq:step_choices}. Let $\eta\in \left( 0,   \frac{cNp_{\min}}{2\bar\varphi\sqrt{p_{\min}}+1} \left(2- \frac{1}{2\chi\vartheta}\right)\right]$ with $p_{\min}:=\min \{p_i\}_{i\in\ca N}$ and $c\in(0,1)$. Then, the sequence $(\bld x(k))_{k\in\bN}$ generated by AD-GENO (Algorithm~\ref{alg:AD-GENO}) converges to the v-GNE of the game in \eqref{eq:game_formulation} almost surely.  \hfill\QEDopen
\end{theorem}
\begin{proof}
See \ref{app:AD_GENO}
\end{proof}

\section{Simulations}
\label{sec:simulations}
We conclude  by proposing two sets of simulations to validate the theoretical results in the previous sections and to highlight the performances of the proposed algorithm.
First, we apply AD-GENO on a network Cournot game and study how delays and different activation sequences affect the convergence. Then, we compare the total computation time required by  AD-GENO, AD-GEED and ADAGNES (in~\cite[Alg.~1]{yi:pavel:2019:asynch_distributed_GNE_w_partial_info}), over different communication graphs.

\subsection{AD-GENO convergence}
\label{sec:sim_AD_GENO_conv}
In a network Cournot game, $N$ firms compete over $m$ markets and the coupling constraints arise from the maximum markets capacities. We consider a smilar formulation to that proposed in~\cite{yu:2017distributed}.
Here, we considered $N=8$ firms, with the possibility to act over $m=3$ markets, i.e., $x_i\in\bR^3$, for all $i\in\ca N$. The local production is bounded in $0\leq x_i \leq \overline x_i$, where each component of $\overline x_i\in\bR^{3}$ is randomly drawn from $[10,45]$. In Figure~\ref{fig:Markets}, the interaction of each firm with the markets is shown, where an edge is drawn between a firm and a market if the one of former's strategies is applied to the latter. Two firms are neighbors if they compete over the same market, therefore the communication network between the firms is the one in Figure~\ref{fig:Network}. 
\begin{figure}[ht]
\centering
\subfloat[][]
{\includegraphics[scale = 0.12]{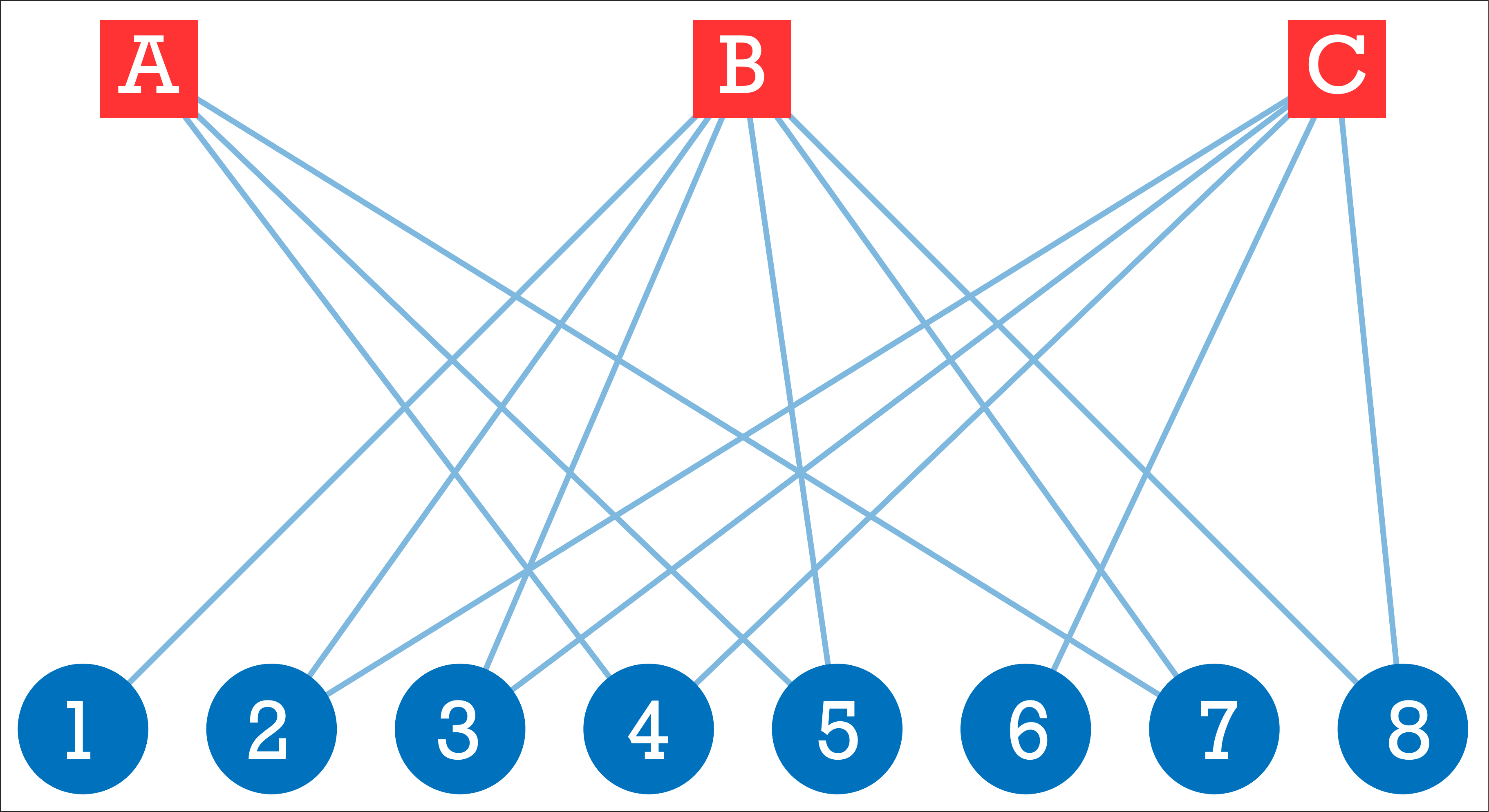}\label{fig:Markets}} \; 
\subfloat[][]
{\includegraphics[scale = 0.10]{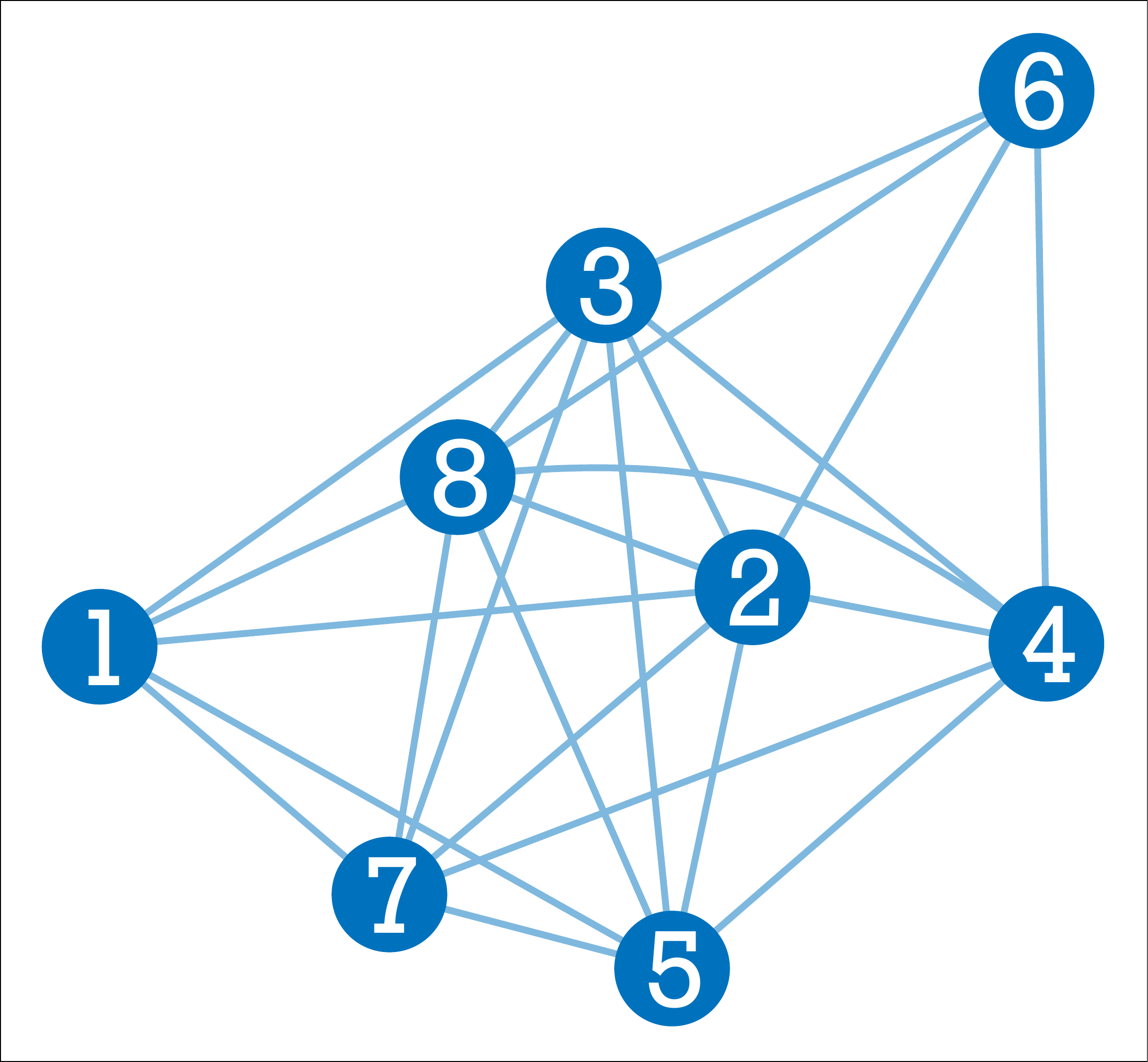}\label{fig:Network}} \\
\caption{(a) Action of players $\{1,\dots,8\}$ over the three markets $A$, $B$, $C$, $D$, (b) Communication network arising from the competition over the markets.}
\label{fig:markets_and_network}
\end{figure}
The coupling constraints are defined by $\bld{Ax}\leq  b$, where $\bld A\coloneqq[A_1,\dots,A_N]\in \bR^{3\times 24}$ while $b\in\bR^{3}$.
The element $[A_i]_{jk}$ is nonzero, if $[x_i]_k>0$ and it is applied to market $j$.
Each nonzero element in  $\bld A$ is randomly chosen from $[0.6,1]$, this value can be seen as the efficiency of a strategy on a market. The components of $b\in\bR^{3}$ are the capacities of the markets, randomly drawn from $[20,100]$. The local cost function is defined as $f_i(x_i,\bld x_{-i})\coloneqq c_i(\bld x) -P(\bld x)^\top A_ix_i$; $c_i(\bld x)$ and it describes  the cost of opting for  a certain strategy, while $P(\bld x)$ is the reward attained. The price is assumed linear $P(\bld x)=\bar P -DA\bld x$, where $\bar P\in\bR^3 $ and $D\in\bR^{3\times 3}$ is a diagonal matrix, their non zero components  are randomly chosen from $[250,500]$ and $[1,5]$ respectively. The function $c_i(\bld x) = x_i^\top Q_i x_i + q_i^\top x_i$ is quadratic, where $Q_i\in\bR^{4\times 4}$ is diagonal and $q_i\in\bR^4$. Their values are randomly chosen from $[1,8]$ and $[1,4]$, respectively.   

\begin{figure}[ht]
\centering
\subfloat[][]
{\includegraphics[width=.4\textwidth]{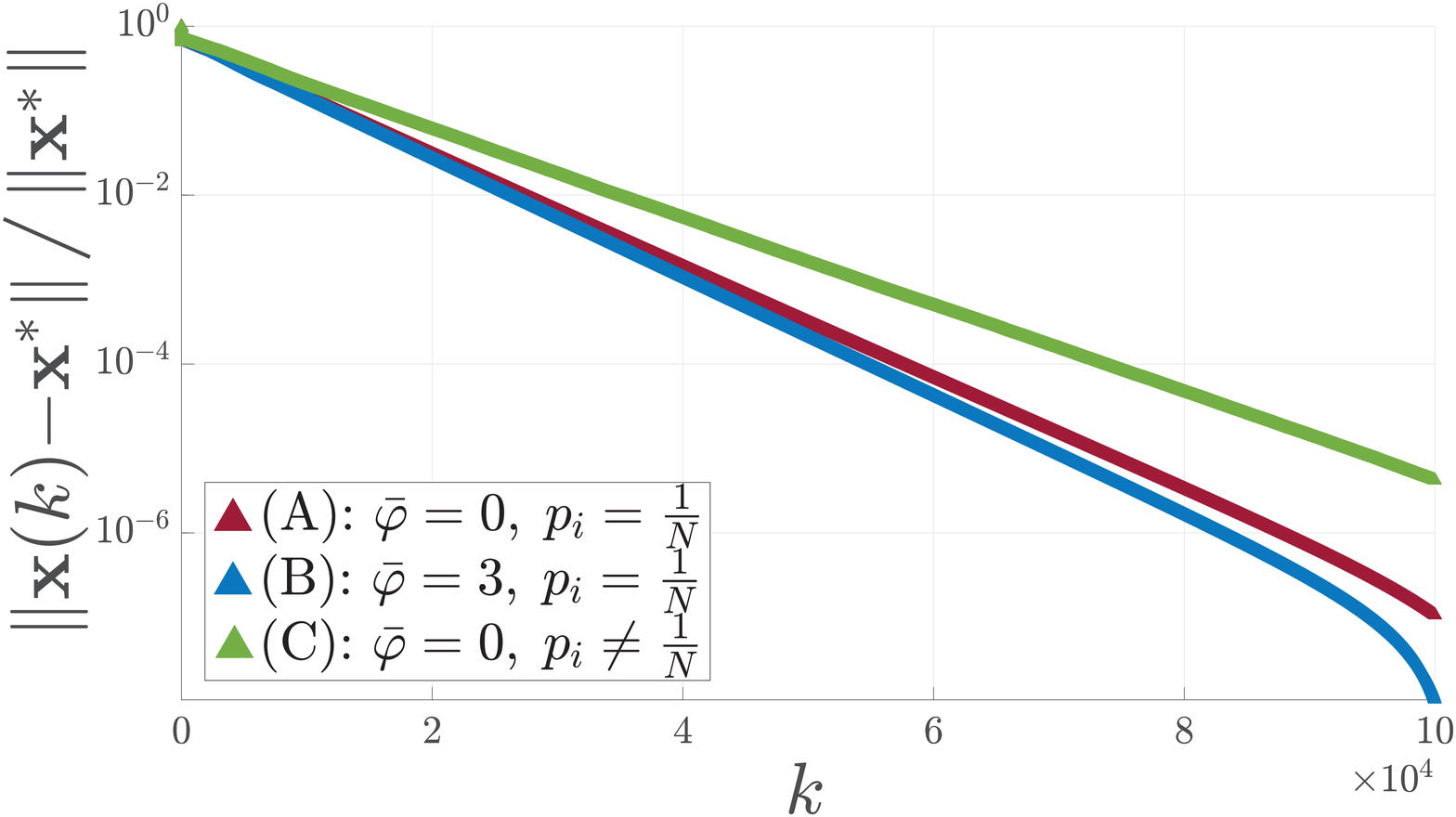}\label{fig:NormXVal}} \\
\subfloat[][]
{\includegraphics[width=.4\textwidth]{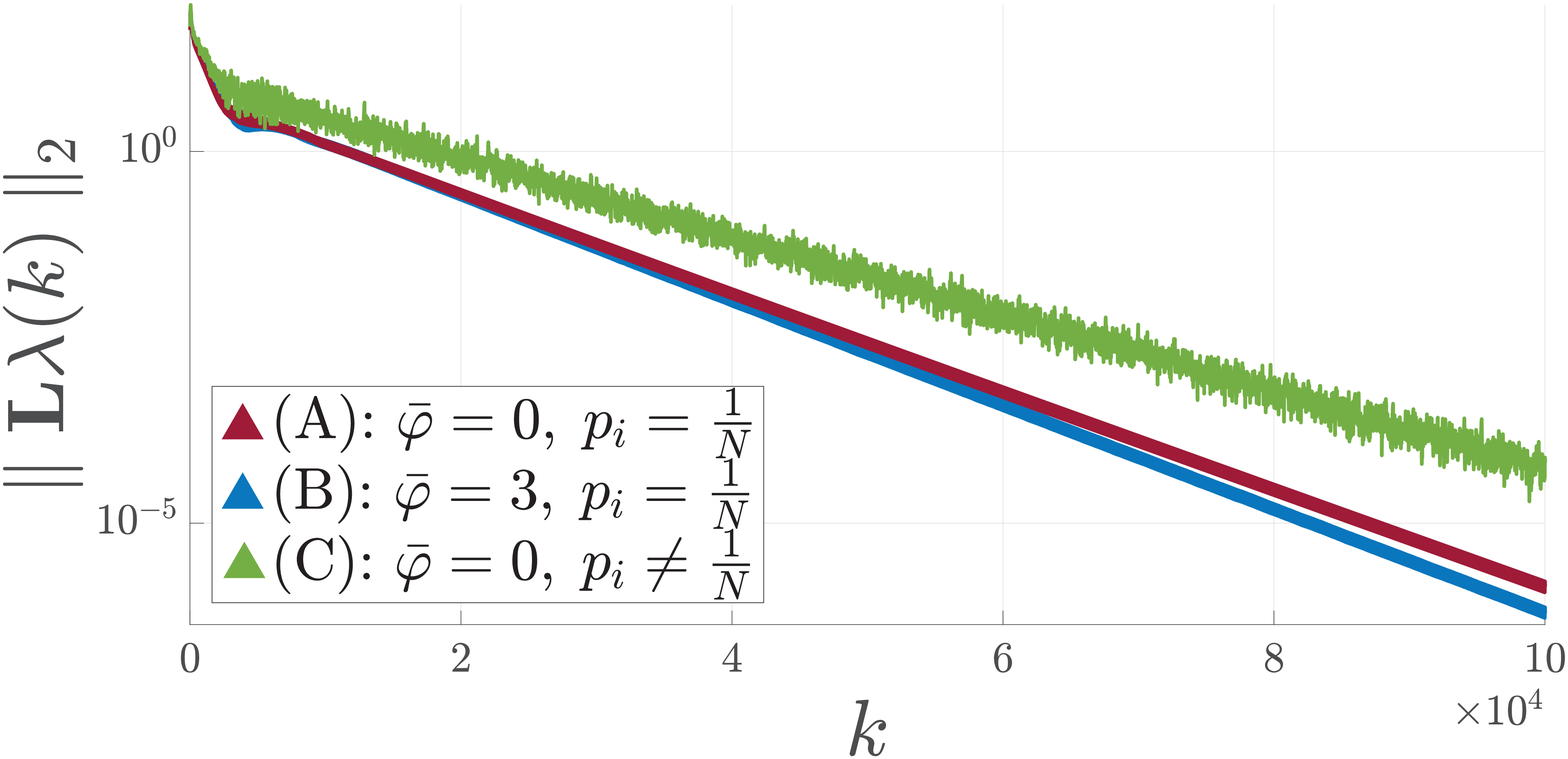}\label{fig:LambdaErr}} \\
\subfloat[][]
{\includegraphics[width=.4\textwidth]{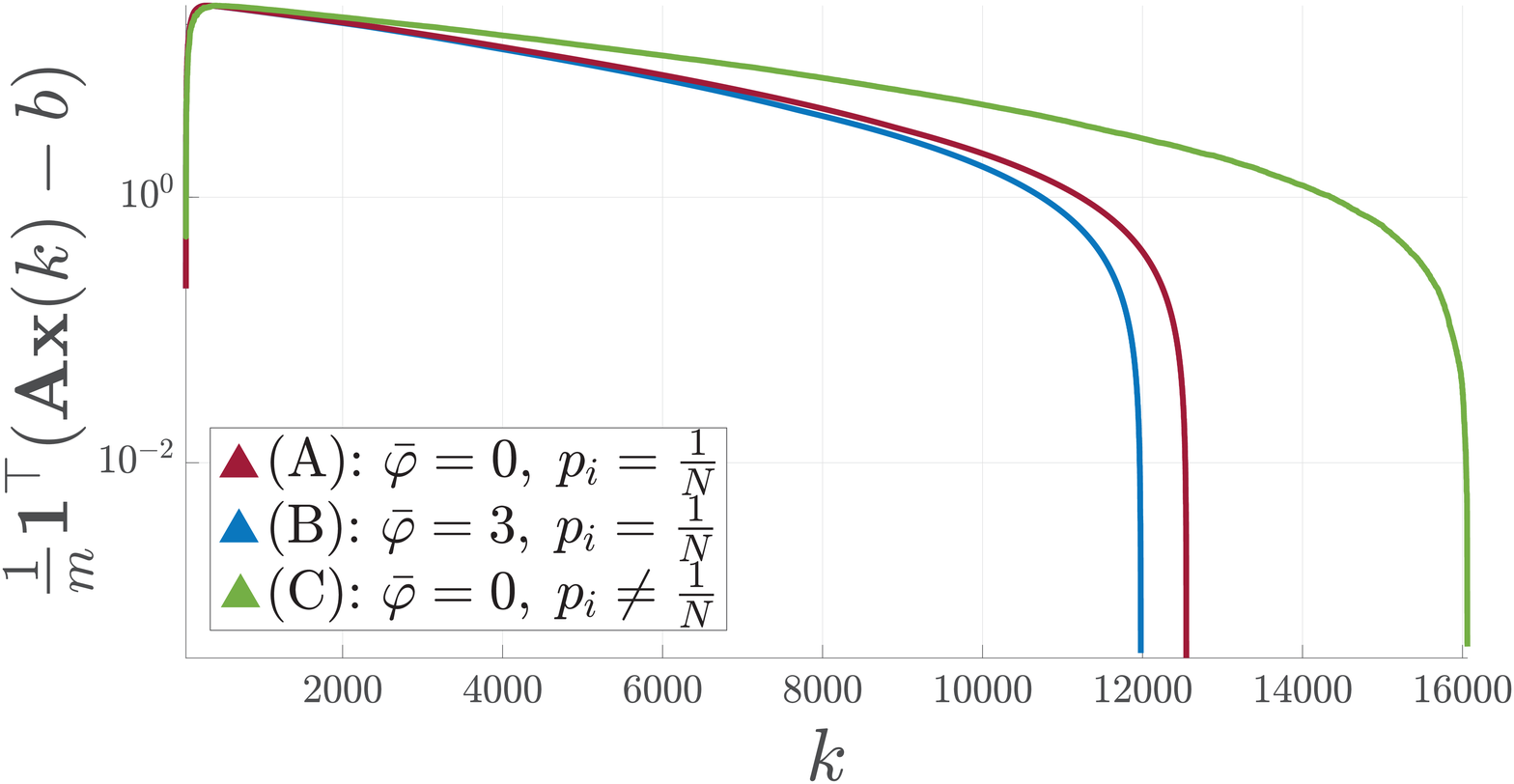}\label{fig:Constr}}\\
\caption{
(a) Normalized distance from the v-GNE, (b) Norm of the disagreement between the dual variables, (c) Constraints violation.}
\label{fig:simulations}
\end{figure}

In order to explore different setups we simulate three different cases:
\begin{enumerate}[(A)]
\item The communication is  delay free ($\bar\varphi>0$) and the activation sequence is alphabetic, and hence $P[\zeta (k)=\bld H_i]=\frac{1}{N}$, for every $i\in\ca N$.
\item The activation sequence is still alphabetic, but the communication may be delayed of $3$ time instants at most, i.e., $\bar\varphi=3$.
\item The communication has no delay, but the probability of update is different   between agents, half of them have $p_i=\frac{1}{6}$, while the rest $p_i=\frac{1}{12}$.
\end{enumerate}

The outcome of these scenarios are presented in Figure~\ref{fig:simulations}. 
The main difference can be noticed in the case of a non-uniform update probability; in fact we notice that a skewer probability  implies  slower converge. From the simulations performed, we have noticed that the presence of delay does not affect the convergence speed in a drastic way, and, in some fortunate cases, can even lead to a faster convergence. 
From simulations, we noticed that the convergence of the dual variables is often  the bottleneck to high convergence performances. In all our algorithms, we mitigated this effect by an appropriate tuning of  $\rho$.  From \eqref{eq:step_choices}, one can deduce that a smaller value of $\rho$ allows to choose a bigger $\varepsilon_i$, thus a bigger step towards consensus. 


\subsection{Comparison between algorithms}
Next, we compare the performance of AD-GENO with respect to AD-GEED and ADAGNES, from a computational time point of view. 
For the comparison with ADAGNES, we consider a modified version of the Nash--Cournot game presented in Section~\ref{sec:sim_AD_GENO_conv} with the coupling equality constraints $\bld A \bld x = b$, namely, the overall production in each market must match the correspondent capacity/demand. Here, we consider $N=40$ firms, each with at most $n_i=2$ products. To provide an extensive comparison, we considered many instances of this game varying the communication between agents, from a complete to a sparse graph. This is achieved by increasing the number of markets. More precisely, we considered the average degree of the nodes to create graphs with a desired sparsity, e.g., if the node average degree is $39$ then the graph is complete. The other quantities in the games are chosen as in the previous section. We compared the algorithms over $160$ different graphs. The computational time required to obtain convergence is compared in the three cases.\footnote{The computation is performed on a single computer, thus the considered time is due to the local updates only and not the communications between the agents.}

\begin{figure}
\centering
\includegraphics[trim= 0 0 0 0, clip, scale = 0.22]{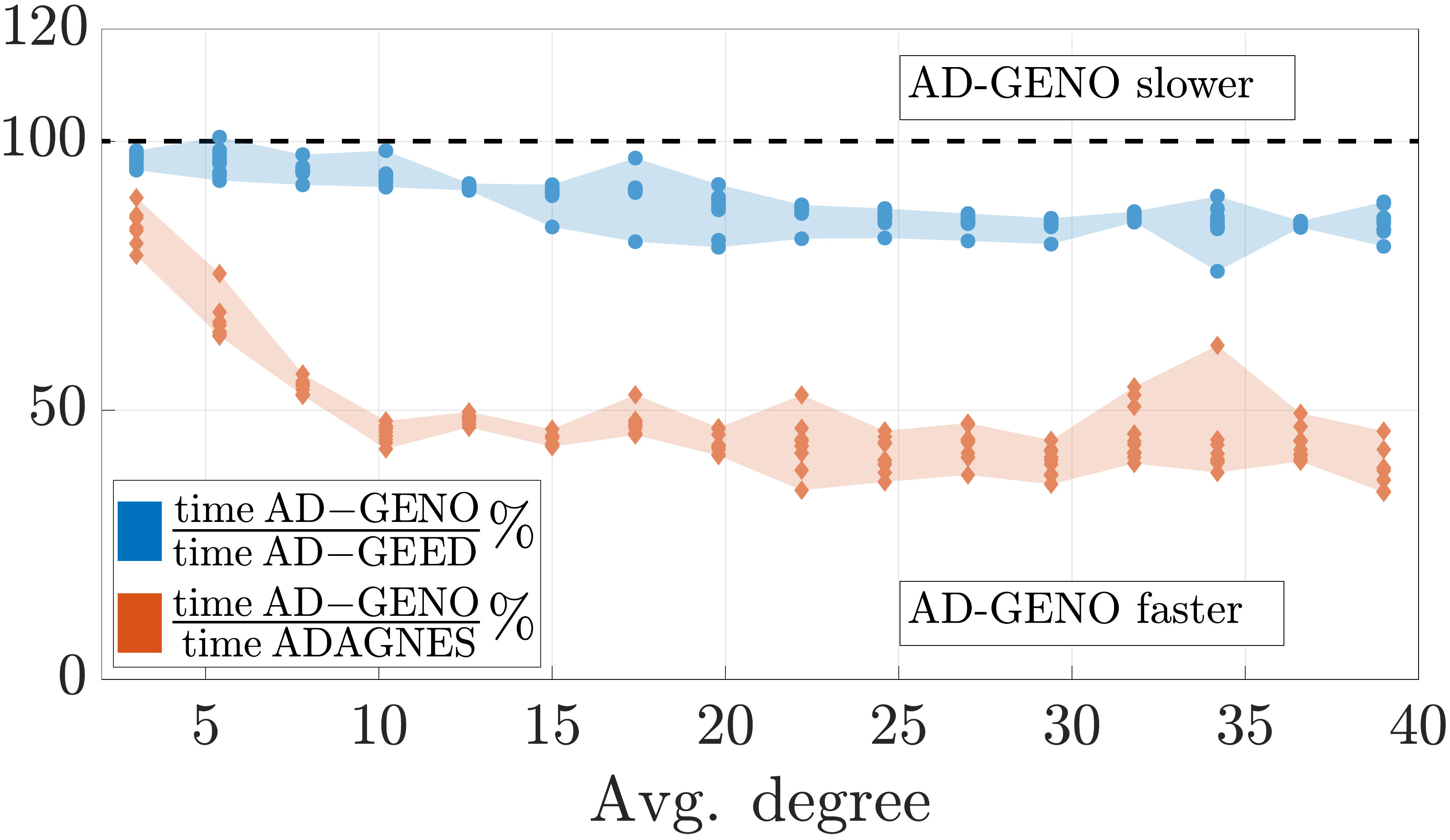}
\caption{Comparison of the computation time of ADAGNES vs AD-GENO (orange diamond) and  AD-GEED vs AD-GENO (blue dots), w.r.t. the variation of the communication network connectivity.
}
\label{fig:alg_cmp}
\end{figure}
The results of the simulations are presented in Figure~\ref{fig:alg_cmp}. As expected AD-GENO always outperforms AD-GEED, since it achieves the same dynamics of $(\bld x,\bld \lambda)$ with fewer auxiliary variables. As expected, the gap between the two algorithms shrinks for a sparse graph while it increases for a dense one, from $\sim 3\%$ to $\sim 20\%$. 
A similar behavior arises when AD-GENO is compared to ADAGNES, due to the increment of auxiliary variables for highly connected graphs. In particular, the advantage in using AD-GENO starts from $\sim 20\%$ when the graph has an average degree of $3$ and becomes  $\sim 60\%$ when the graph is complete.


\section{Conclusion}
\label{sec:conclusion}
Solving a GNE seeking problem in strongly monotone games is possible via AD-GENO in an asynchronous fashion, with node variables only, and by ensuring resilience to delayed information. In our numerical experience, AD-GENO outperforms the available solutions in the literature, both in terms of computational time and number of variables required. 


Unfortunately, the  ``ARock'' framework does not ensure robustness to lossy communication. This is currently an open problem that is left to future research. 
Another interesting topic is the generalization of the algorithm to the case of time-varying communication networks, as the independence from the edge variables makes this approach more suitable to address this problem.

\section*{Funding}
The work of Cenedese, and Cao was supported in part by the European Research Council (ERC-CoG-771687) and the Netherlands Organization for Scientific Research (NWO-vidi-14134). The work of Grammatico was partially supported by NWO under research project OMEGA (613.001.702) and P2P-TALES (647.003.003) and by the European Research Council under research project COSMOS (802348).
\appendix 

\section{Proofs of Section~\ref{sec:synch_case}}
\label{app:proof_synch_case}

\subsection{Proof of Lemma~\ref{lem:max_mon_of_operators}}
%
The operator $\ca A$ is the sum of two operators $\ca A_1+\ca A_2$, the first is a real skew symmetric matrix, hence maximally monotone by \cite[Ex.~20.35]{bauschke2011convex}. $\ca A_2$ is $N_{\bld \Omega} \times 0 \times N_{\bR^{mN}_{\geq 0}} $, thus maximally monotone from \cite[Ex.~20.26]{bauschke2011convex}. 
By \cite[Cor.~25.5(i)]{bauschke2011convex}, we conclude that $\ca A$ is maximally monotone, since $\mathrm{dom}(\ca A_1) = \bR^{n+(E+N)m}$.

Notice that the cocoercity of an operator also implies its maximally monotoniticy, see \cite[Ex.~20.31]{bauschke2011convex}.
Next, we prove the cocoercivity of $\ca B$. For all $\varpi_1,\varpi_2\in \mathrm{dom}(\ca B)$ , it holds
\begin{equation*}
\begin{split}
&\langle \ca B(\varpi_1)-\ca B(\varpi_2) |x-y\rangle  = <\left[\begin{smallmatrix} F(\bld x_1) - F(\bld x_2) \\
0\\
\bld{L(\lambda_1-\lambda_2)} \end{smallmatrix}\right] | \varpi_1-\varpi_2>\\
&= <F(\bld x_1) - F(\bld x_2)| \bld x_1-\bld x_2>+ (\bld \lambda_1-\bld \lambda_2)^\top\bld{L}(\bld\lambda_1-\bld\lambda_2) \\
& \geq \tfrac{\alpha}{\ell^2} \lVert F(\bld x_1) - F(\bld x_2)  \rVert^2 + \lambda_{\max}(\bld L)^{-1} \lVert \bld{L}(\bld\lambda_1-\bld\lambda_2)  \rVert^2\\
&\geq \min\big\{\tfrac{\alpha}{\ell^2},\lambda_{\max}(\bld L)^{-1} \big\} \lVert \ca B(\varpi_1)-\ca B(\varpi_2) \rVert^2 \\
& = \chi \,\lVert \ca B(\varpi_1)-\ca B(\varpi_2) \rVert^2  \:,
\end{split}
\end{equation*}
where the first inequality is attained by the $\alpha$-strong monotonicity and $\ell$ Lipschitzianity of $F$ and by recalling that for every symmetric real matrix $M\succeq 0$ it holds $y^\top My\geq \lambda_{max}(M)^{-1} \lVert My\rVert^2$, with $\lambda_{\max}(M)\in\bR_{>0}$. \hfill\QED
\smallskip

\subsection{Proof of Theorem~\ref{th:eq_are_vGNE_mod_map}}
Consider the equilibrium point $\col (\bld x^*,\bld \lambda^*,\bld z^*)$ with $\bld z^*=\bld{V^\top \sigma}^*$. Then,~\eqref{eq:row2_sync_z} at the equilibrium reduces to $\bld 0 = -\rho\bld L\bld \lambda^*$, and thus $\bld \lambda^* = \lambda^*\otimes \bld 1$. 

Moreover, manipulating \eqref{eq:row3_sync_z} and evaluating it in the equilibrium lead to  
\begin{equation}
\label{eq:row2_proof_equilibrium}
\bld 0 \in N_{\bR^{mN}_{\geq 0}}(\bld \lambda^*)+\bar b-\Lambda \bld x^*+\rho\bld z^* + (2\rho \delta +1) \bld L \bld \lambda^*  \,.
\end{equation}
Recalling that $\bld L \bld \lambda^* = \bld 0$ and multiplying both sides of  \eqref{eq:row2_proof_equilibrium} by $(\bld 1^\top \otimes I)$ leads to
\begin{equation}
\label{eq:row2_proof_equilibrium_2}
\bld 0 \in  (\bld 1^\top \otimes I)(N_{\bR^{mN}_{\geq 0}}(\bld \lambda^*)+\bar b-\Lambda \bld x^* +\rho\bld z^* )\,.
\end{equation}   
Using the fact that $\sum_{i\in\ca N}N_{\bR^m_{\geq 0}} (\lambda^*)=N_{\cap_{i\in\ca N} \bR^m_{\geq 0}} (\lambda^*)=N_{\bR^m_{\geq 0}} (\lambda^*)$ and the assumption $\bld 1^\top \bld z^* =0$,  \eqref{eq:row2_proof_equilibrium_2} becomes
\begin{equation}
\label{eq:row2_proof_equilibrium_3}
\bld 0 \in  N_{\bR^{m}_{\geq 0}}(  \lambda^*)+\bar b - A \bld x^*\,.
\end{equation}

Finally, \eqref{eq:row3_sync} evaluated in the equilibrium is 
\begin{equation}
\label{eq:row3_proof_equilibrium_2}
\bld 0 \in F(\bld x^*)+N_{\bld \Omega}(\bld x^*)+\Lambda^\top \bld \lambda^*\\
\,,
\end{equation}   
or equivalently
\begin{equation}
\label{eq:row3_proof_equilibrium_3}
\bld 0 \in  \nabla f_i(x_i^*,\bld x^*_{-i})+N_{\Omega_i}( x^*_i)+A_i^\top  \lambda^*\:,\quad \forall i\in\ca N\,.
\end{equation}
Inclusions \eqref{eq:row3_proof_equilibrium_3} and \eqref{eq:row2_proof_equilibrium_3} are the KKT conditions in \eqref{eq:KKT_VI}, and hence from \cite[Th.~3.1]{facchinei:fischer:piccialli:07} we conclude that $\col (\bld x^*,\bld \lambda^*,\bld z^*)$ is a v-GNE of the game.
\hfill\QED
\smallskip

\subsection{Proof of Theorem~\ref{th:convergence_sync}}
From \cite[Th.~2]{feingold:varga:1962:gerschgorin_circle}, the choice of $\vartheta$, $ \varepsilon_i$, $\delta$ and $\tau_i$  in~\eqref{eq:step_choices} implies that $\Phi-\vartheta I\succ 0$. Then, invoking \cite[Lem.~5.6]{pavel2017:distributed_primal-dual_alg} we obtain that $\Phi^{-1}\ca B$ and $\Phi^{-1}\ca A$ are respectively $\chi\vartheta$-cocoercive and  maximally monotone in the $\Phi$ induced norm. Furthermore, it also shows that $(\Id-\Phi^{-1}\ca B)$ is $\frac{1}{2\chi\vartheta}$-AVG and $\mathrm{J}_{\Phi^{-1}\ca A}\coloneqq(\Id+\Phi^{-1}\ca A)$ is FNE. Applying \cite[Prop.~2.4]{combettes:yamada:15}, we conclude that $T$ is $\frac{2\chi\vartheta}{4\chi\vartheta-1}$-AVG. The Krasnosel'ski\u i iteration in \eqref{eq:Krasno_iter_synch} converges to $\varpi^*\in\fix(T)$ if $\eta\in\big(0,\frac{4\chi\vartheta-1}{2\chi\vartheta}\big)$, \cite[Th.~5.14]{bauschke2011convex}. 

The above argument establishes that $\lim_{k\rightarrow +\infty}\bld \sigma^k = \bld \sigma^*$, and hence $\lim_{k\rightarrow +\infty} \bld V^\top \bld \sigma^k = \bld V^\top \bld \sigma^*=:\bld z^*$. Therefore, $\bld z$ converges and consequently we conclude that Algorithm~\ref{alg:synch_alg} converges to $\col(\bld x^*,\bld \lambda^*,\bld z^*)$. The choice of the initial value $\bld z_0=\bld 0$, implies that $\bld 1^\top \bld z^k =0$, for all $k\in\bN$ since its values will be in the range of $L$.
 Finally, applying Theorem~\ref{th:eq_are_vGNE_mod_map} we  prove that the equilibrium is the v-GNE of the original game.
\hfill\QED
\smallskip

\section{Proof of Theorem~\ref{th:convergence_asynch_sigma}}
\label{app:prrof_AD_GEED}
From the proof of Theorem~\ref{th:convergence_sync}, we know that $T$ is $\frac{2\chi\vartheta}{4\chi\vartheta-1}$-AVG, and therefore it can be rewritten as $T=(1-\frac{2\chi\vartheta}{4\chi\vartheta-1})\Id+\frac{2\chi\vartheta}{4\chi\vartheta-1} P$, where $P$ is nonexpansive, \cite[Prop.~4.35]{bauschke2011convex}. By substituting it into \eqref{eq:Krasno_asynch_final}, we obtain
\begin{equation}\label{eq:Krasno_asynch_rule_proof}
\varpi^{k+1}=\varpi^k+ \frac{2\eta\chi\vartheta}{4\chi\vartheta-1} \zeta^k( P-\Id)\hat\varpi^k\,.
\end{equation}
For \eqref{eq:Krasno_asynch_rule_proof}, we apply \cite[Lem.~13 and Lem.~14]{peng2016arock}, to conclude that $(\varpi^k)_{k\in\bN}$ is bounded and that it converges almost surely to $\varpi^*\in\fix(P)=\fix(T)$, for $\frac{2\eta\chi\vartheta}{4\chi\vartheta-1}\in(0,\frac{cNp_{\min}}{2\bar\varphi\sqrt{p_{\min}}+1}]$, thus if $\eta\in \left( 0,   \frac{cNp_{\min}}{2\bar\varphi\sqrt{p_{\min}}+1} \left(2- \frac{1}{2\chi\vartheta}\right)\right]$. 
Since $\fix(T)=\zer(\ca A+\ca B)$, we conclude from Proposition~\ref{prop:zer_AB_are_vGNE} that $\{\bld x^k\}_{k\in\bN_{\geq 0}}$ converges to the v-GNE of \eqref{eq:game_formulation} almost surely.
\hfill\QED
\smallskip

\section{Proof of Theorem~\ref{th:convergence_AD-GENO}}
\label{app:AD_GENO}
The change of auxiliary variables in Algorithm~\ref{alg:AD-GENO} from $\bld \sigma$ to $\bld z$ leads to a different update rule for $\bld \lambda$, while the one for $\bld x$ remains unchanged. Therefore, if we show that the modified update of $\bld \lambda$ is equivalent to the one in Algorithm~\ref{alg:AD-GEED}, we can infer the convergence from Theorem~\ref{th:convergence_asynch_sigma}.

\smallskip
We prove by induction that, given an agent $i$, the update of $\lambda_i$ at time $k$ in Algorithm~\ref{alg:AD-GEED}~and~\ref{alg:AD-GENO} are equivalent. Note that the two update rules are equivalent if it holds that 
\begin{equation}\label{eq:condition_z_sigma}
\tilde z_{i}(k)=([V^\top]_{i}\otimes I_m)\hat{ \bld \sigma}(k)\:,
\end{equation} 
for every $k>0$ and $i\in\ca N$.

\underline{\textit{Base case:}} Iteration $k$ is the first in which agent $i$ updates its variables. 
If $k=0$, then in AD-GENO $\mu_i =\bld 0_{m}$ for every $i$ and $\hat{\bld \sigma}(0) = \bld 0_{mM}$ in AD-GEED, hence \eqref{eq:condition_z_sigma} is trivially verified.

If instead $k>0$, it holds that $z_{i}(k)=z_{i}(0)=\bld 0_m$, while $\hat{\bld \sigma}(k) \not = \bld 0_{mM}$, since the neighbors of $i$ can update more than once before the first update of $i$ (as shown in Example~\ref{ex:sigma_z}). We define for each $j\in\ca N_i$ the set $\ca S_{j}(k)$, a $t\in\bN$ where $t<k$ belongs to $\ca S_{j}(k)$ if, at the iteration $t$, the agent $j$ completes an update. 
The maximum time in $\ca S_{j}(k)$ is denoted as $m_j(k):=\max\{\ca S_{j}(k)\}$ and $\check {\ca S}_{j}(k):={\ca S}_{j}(k)\setminus m_j(k)$.
From this definitions, we obtain that   
\begin{equation}\label{eq:E_top_sigma_hat}
([V^\top]_{i}\otimes I_m)\hat{ \bld \sigma}(k) = \textstyle{ \sum_{l\in\ca E^{\mathrm{out}}_i} } \sigma_{l}(0) - \textstyle{ \sum_{d\in\ca E^{\mathrm{in}}_i} } \sigma_{d}(m_j(k))\:,
\end{equation}      
where $j$ is the element of $e_d$ different from $i$. Furthermore, from the update rule of $\sigma_d$ in Algorithm~\ref{alg:AD-GEED}, we derive
\begin{equation}\label{eq:definition_sigma_d}
\begin{split}
 \sigma_{d}(m_j(k)) &= \sigma_{d}\left(\max\{\check{\ca S}_{j}(k)\}\right)\\
 & \hspace{0.5cm} + \eta\delta\rho \bigg(\lambda_{j}\left(\max\{\check{\ca S}_{j}(k)\}\right)-\lambda_{i}(0)\bigg)\\
 & = \eta\delta\rho\: \textstyle{ \sum_{h\in\check{\ca S}_{j}(k) } }\left(   \lambda_{j}(h) - \lambda_{i}(0) \right) 
\end{split}
\end{equation}
Substituting  \eqref{eq:definition_sigma_d} into \eqref{eq:E_top_sigma_hat} leads to
\begin{equation}\label{eq:E_top_sigma_hat2}
\begin{split}
([V^\top]_{i}\otimes I_m)\hat{ \bld \sigma}(k) &= \rho\eta \delta\left( \textstyle{\sum_{j\in\ca N_i\setminus \{i \}} }\left|\check{S}_{j}(k)\right| \lambda_{i}(0) \right. \\ & \left.\quad  -  \textstyle{\sum_{j\in\ca N_i\setminus \{i \}} \sum_{h\in\check{\ca S}_{j}(k)} }\lambda_{j}(h) \right)\,.\\
\end{split}
\end{equation}
From the definition given in Algorithm~\ref{alg:AD-GENO} of $\mu_i$, we attain that $([V^\top]_{i}\otimes I_m)\hat{ \bld \sigma}(k)=\eta \delta\mu_i = \tilde z_i(k) $, therefore \eqref{eq:condition_z_sigma} hold.

\smallskip
\underline{\textit{Induction step:}} Suppose that \eqref{eq:condition_z_sigma} holds for some $\bar k >0 $ that corresponds to the latest iteration in which agent $i$ performed the update, i.e. $z_{i}(\bar k)\not = \bld 0$. 

Consider the next iteration $k$ in which agent $i$ updates, $k>\bar k$. Here, $\ca S_{j}(k)$ is defined as above, but for time indexes $(\bar k,k]$. Following similar reasoning in the previous case, we obtain
\begin{equation}\label{eq:E_top_sigma_hat_induction}
\begin{split}
([V^\top]_{i}&\otimes I_m)\hat{ \bld \sigma}(k) =([V^\top]_{i}\otimes I_m)\hat{ \bld \sigma}(\bar k) \\
&\quad +\eta\delta\rho\textstyle{\sum_{l\in\ca E_i^{\mathrm{out}}}([V]_{l}\otimes I_m)\hat{ \bld \lambda}(\bar k) } \\
&\quad+\eta \delta\rho \left(  \textstyle{ \sum_{j\in\ca N_i\setminus \{i \}} }\left|\check{S}_{j}(k)\right| \lambda_{i}(0) \right.\\
&\quad\left.-  \textstyle{\sum_{j\in\ca N_i\setminus \{i \}} \sum_{h\in\check{\ca S}_{j}(k)} }\lambda_{j}(h) \right)\\
\end{split}
\end{equation}
where we used the fact that $l\in\ca E_i^{\mathrm{out}}$ is updated at the same time of $i$. Furthermore, from the induction assumption,
\begin{equation}
\begin{split}
([V^\top]_{i}\otimes I_m)\hat{ \bld \sigma}(k) &=   z_i (\bar k)+\eta \delta\rho \left( \textstyle{\sum_{j\in\ca N_i\setminus \{i \}} } \left|\check{S}_{j}(k)\right| \lambda_{i}(0) \right. \\
&\qquad  \left. -  \textstyle{\sum_{j\in\ca N_i\setminus \{i \}} \sum_{h\in\check{\ca S}_{j}(k)} } \lambda_{j}(h) \right)\\
& = z_i(\bar k) +\eta \delta \rho\mu_i= \tilde z_i(k)\:,
\end{split}
\end{equation}
where the last step holds because in the reading phase of Algorithm~\ref{alg:AD-GENO}, we reset to zero the values of $\mu_i$, every time that $i$ starts an update. 
Therefore, \eqref{eq:condition_z_sigma} holds for $k$. 

This concludes the proof by induction, showing that the update of $\lambda_{i}$ defined in Algorithm~\ref{alg:AD-GENO} is equivalent to the one in Algorithm~\ref{alg:AD-GEED}, for any $k$. Finally, since the pair $(\bld x,\bld \lambda)$ has an update rule equivalent to the one in AD-GEED, the convergence of $(\bld x(k))_{k\in\bN}$ to the v-GNE of the game \eqref{eq:game_formulation} follows from Theorem~\ref{th:convergence_asynch_sigma}.     
\hfill\QED 



\bibliographystyle{elsarticle-num} 

\balance
\bibliography{libraryCC,librarySG}

\end{document}
\endinput